\documentclass[12pt, draftclsnofoot, onecolumn]{IEEEtran}
\usepackage{amsmath,amsfonts}
\usepackage{algorithmic}
\usepackage{algorithm}
\usepackage{amsthm}
\usepackage{array}
\usepackage[caption=false,font=normalsize,labelfont=sf,textfont=sf]{subfig}
\usepackage{textcomp}
\usepackage{stfloats}
\usepackage{bm}
\usepackage{url}
\usepackage{verbatim}
\usepackage{graphicx}
\usepackage{amssymb}
\usepackage{cite}
\usepackage{arydshln}
\usepackage[colorlinks, citecolor=black, linkcolor=black]{hyperref}
\newcommand{\overbar}[1]{\mkern 2mu\overline{\mkern-2mu#1\mkern-2mu}\mkern 2mu}

\usepackage{xcolor}

\newtheorem{proposition}{Proposition}

\newtheorem{corollary}{Corollary}

\newcommand{\mb}[1]{\mathbf{#1}}
\newcommand{\mc}[1]{\mathcal{#1}}
\newcommand{\mbb}[1]{\mathbb{#1}}

\newcommand{\wt}[1]{\widetilde{#1}}
\newcommand{\wh}[1]{\widehat{#1}}
\newcommand{\ob}[1]{\overbar{#1}}

\begin{document}

\title{\huge{On the Uplink Distributed Detection in UAV-enabled Aerial Cell-Free mMIMO Systems}}

\author{Xuesong Pan, Zhong Zheng,~\IEEEmembership{Member, IEEE}, Xueqing Huang,~\IEEEmembership{Member, IEEE}, Zesong Fei,~\IEEEmembership{Senior Member, IEEE}
\thanks{X. Pan, Z. Zheng and Z. Fei are with the School of Information and Electronics, Beijing Institute of Technology, Beijing 100081, China (e-mail: \{xs.pan, zhong.zheng, feizesong\}@bit.edu.cn).}
\thanks{X. Huang is with the Department of Computer Science, New York Institute of Technology, NY 11568, USA (e-mail: xhuang25@nyit.edu).}}

\markboth{Submitted to XXXXXXXX}%
{}

\maketitle

\begin{abstract}	
In this paper, we investigate the uplink signal detection approaches in the cell-free massive MIMO systems with unmanned aerial vehicles (UAVs) serving as aerial access points (APs). The ground users are equipped with multiple antennas and the ground-to-air propagation channels are subject to correlated Rician fading. To overcome huge signaling overhead in the fully-centralized detection, we propose a two-layer distributed uplink detection scheme, where the uplink signals are first detected in the AP-UAVs by using the minimum mean-squared error (MMSE) detector depending on local channel state information (CSI), and then collected and weighted combined at the CPU-UAV to obtain the refined detection. By using the operator-valued free probability theory, the asymptotic expressions of the combining weights are obtained, which only depend on the statistical CSI and show excellent accuracy. Based on the proposed distributed scheme, we further investigate the impacts of different distributed deployments on the achieved spectral efficiency (SE). Numerical results show that in urban and dense urban environments, it is more beneficial to deploy more AP-UAVs to achieve higher SE. On the other hand, in suburban environment, an optimal ratio between the number of deployed UAVs and the number of antennas per UAV exists to maximize the SE.

\end{abstract}


\section{Introduction} \label{intro}
Unmanned aerial vehicles (UAVs) have been found in a wide range of applications in wireless communication during the past few decades, especially serving as aerial base stations to expand the network coverage and enhance communication quality~\cite{mozaffari2021toward}. Nevertheless, due to the limited carrying capability, the UAV-based stations can be only equipped with a few antennas and it's difficult to achieve satisfactory communication quality with a single UAV-based station. Therefore, it is interesting to investigate the deployment of a swarm of UAV-based stations, which collaborate in signal transmission and processing~\cite{geraci2022will}. 


Parallel to the research on UAV communications, a novel network architecture called the cell-free massive multi-input multi-output (CF mMIMO) system has been proposed to fulfill the demand on the rapid growth of data throughput~\cite{ngo2017cell}. In the CF mMIMO system, a large number of access points (APs) are deployed and connected to a central processing unit (CPU) via the fronthaul link. Instead of creating several autonomous cells, the geographically distributed APs jointly serve the UEs by coherent transmission and reception. By exploiting the increased spatial degree of freedom due to the cooperation, CF mMIMO increases the number of simultaneously served users as well as improves the achieved sum rate, compared to the conventional noncooperative cellular systems~\cite{bjornson2020scalable}. Meanwhile, the CF mMIMO can provide significant performance improvement compared to the small-cell system~\cite{nayebi2017precoding, ngo2017cell} and more uniform communication quality to the edge UEs. Due to these advantages, it is of great interest to empower the UAV swarm with CF mMIMO capability to compensate for the limited onboard hardware resources, which could enable high-speed communications towards massive numbers of ground UEs. Moreover, benefiting from the higher altitude and flexible deployment, the UAV-based stations are more likely to establish Line-of-Sight (LoS) links with ground UEs and provide on-demand services~\cite{zheng2021uav}. As shown in Fig.~\ref{figModel}, the UAVs are divided into the access point UAVs (AP-UAVs) for local signal transmission and reception, as well as the CPU-UAV for joint signal processing. The AP-UAVs and the CPU-UAV are inter-connected via wireless fronthaul links. 

\begin{figure}[ht] 
	\centerline{\includegraphics[width=0.6\columnwidth]{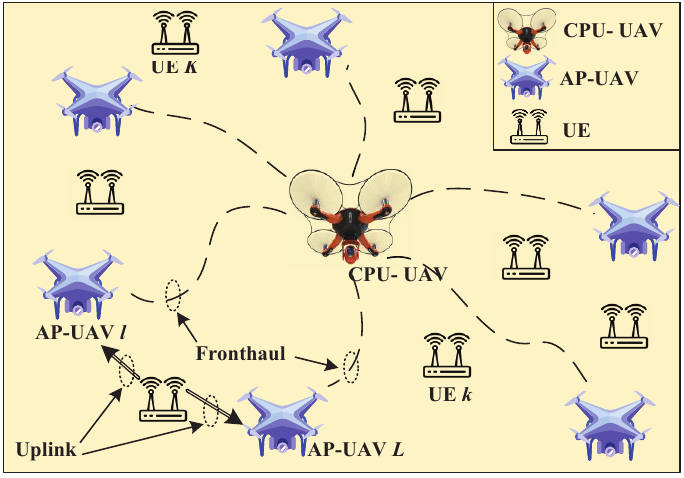}}
	\caption{Multi-UAVs based distributed receiving with one-shot combining.}
	\label{figModel}
\end{figure}

The integration of UAV communications and CF mMIMO was first proposed by Carmen et al.~\cite{d2019cell}, where UAVs are treated as aerial UEs and APs are deployed on the ground. The aerial CF mMIMO system with AP-UAVs was studied by Carles et al., where the UAVs serve as aerial APs and are equipped with a single antenna~\cite{diaz2022cell}. In \cite{wang2022cell}, the downlink communication between the AP-UAVs and ground UEs was studied and the location of AP-UAVs was optimized to maximize the sum rate. Carles et al. aimed to increase the spectral efficiency (SE) by maximizing the minimum local-average signal-to-interference-plus-noise ratio (SINR) in uplink CF mMIMO system, where UAVs are deployed as aerial APs~\cite{diaz2021deployment}. However, the above literature all considers that the UAVs are equipped with a single antenna, which can't exploit the spatial diversity of multiple antennas.

Despite the potential performance enhancement of the UAV-based CF mMIMO systems, this network architecture requires the network-wide CSI available at the CPU-UAV to optimally design the transceivers at the AP-UAVs. However, the interaction of CSI and data between AP-UAVs and CPU-UAV results in immense interactive overhead and unpredictable delay, which is especially critical in UAV-based networks due to wireless fronthaul.

\subsection{Related Works}
 In order to decrease the interactive overhead, network-centric and user-centric approaches are proposed in \cite{zhang2009networked} and \cite{bjornson2011optimality}, respectively, both of which divide the UEs or APs into several subsets. The interactive overhead is significantly reduced by limiting the interaction within the subsets, but leading to notable performance loss~\cite{garcia2010dynamic}. To utilize the macro-diversity as well as decrease the interaction overhead, the large-scale fading decoding (LSFD) strategy was first proposed by Ansuman et al. to reduce interference in cellular system~\cite{adhikary2017uplink} and then applied to CF mMIMO systems in~\cite{bjornson2019making}. In the LSFD strategy, the APs first detect the received signals using the local channel estimates and then forward the detected data to the CPU, where the final joint detection is performed. The key idea of LSFD is that the final detection only depends on the large-scale fading coefficients, which does not need the instantaneous interaction between CPU and APs. 
 
 Based on the same structure, a plethora of papers on LSFD have been published in recent years. Elina et al. investigated the uplink performance of CF mMIMO systems with LSFD receiver, where single-antenna APs first estimate the signal from the single-antenna UEs using the matched filter (MF) with local CSI, and the CPU computes the LSFD based on large-scale fading coefficients~\cite{nayebi2016performance}. Zhang et al. considered the local zero-forcing (ZF) detector with LSFD over Rayleigh fading channels, where different variants of ZF were proposed to achieve high SE or reduce computational complexity~\cite{zhang2021local}. Furthermore, Emil et al. proposed four different levels of cooperation in CF mMIMO systems, including fully-centralized processing, local processing with LSFD, local processing with simple centralized decoding, and small-cell network~\cite{bjornson2019making}. The local processing in~\cite{bjornson2019making} is implemented in APs with MMSE detector, which trade-offs between signal improvement and interference suppression. In the above papers, the ground UEs are assumed to be equipped with a single antenna, while the contemporary UEs have already equipped with multiple antennas. 
 
 With the emergence of new applications in 6G wireless networks, it's necessary to equip the UEs with multi-antennas to improve system performance. Trang C et al. investigated the uplink and downlink SE of CF mMIMO systems with multi-antennas UEs over uncorrelated Rayleigh fading channels, where the closed-form expressions of both uplink and downlink SE are derived~\cite{mai2019uplink, mai2020downlink}. Based on \cite{bjornson2019making}, Wang et al. studied the uplink performance of CF mMIMO system with multi-antennas UEs, where four levels of cooperation are considered over correlated Rayleigh fading channels~\cite{wang2022uplink}. In addition, based on the same system structure, uplink precoding of multi-antennas UEs was investigated to improve uplink SE~\cite{wang2022iteratively}.
 
 However, most existing works, e.g. \cite{nayebi2016performance, zhang2021local, bjornson2019making, mai2019uplink, mai2020downlink, wang2022uplink, wang2022iteratively}, only considered the Non-Line-of-Sight (NLoS) connection between the APs and UEs and ignored the existence of the LoS path, which is the dominant channel component in the UAV-based CF mMIMO systems. In addition, the calculation of the combining weights in LSFD is given in a semi-closed form manner and requires time-consuming numerical simulations, which is difficult to practical application~\cite{wang2022uplink}.

\subsection{Contributions}
To address the above issues, a general UAV-based CF mMIMO system is considered with multi-antennas equipped at both APs and ground UEs, while communicating over generalized Rician fading channels. The main contributions of this paper are summarized as follows: 
\begin{itemize}
	\item We consider a CF mMIMO system with UAVs being the aerial APs and multi-antennas UEs on the ground. In addition, we adopt the Rician fading channel with Kronecker’s correlation model and investigate the pilot contamination over the Rician MIMO channel. By taking the LoS path into account, the considered channel model is more practical and suitable for the UAV-based CF mMIMO system.
	
	\item We design a two-layer detection scheme consisting of local MMSE detection at AP-UAVs and one-shot combining at the CPU-UAV. By using the operator-valued free probability theory, the asymptotically optimal combining weights are derived for the large dimensional MIMO systems, which only depend on the channel statistics and are asymptotically accurate. The free probability theory takes full advantage of the invariance properties of the propagation channels and is a convenient tool to study mutual information, information plus noise models, etc~\cite{couillet2011random}.
	
	\item We investigate the impact of the network configuration options on the achieved sum rate in various propagation environments~\cite{alzenad20173}. Specifically, given a fixed total number of antennas across AP-UAVs, a.k.a. fixing the total transceivers’ capability at the network side, the tradeoff between the deployed number of AP-UAVs ($L$) and the number of equipped antennas per each AP-UAV ($M$) is shown by numerical simulations. Results show that there exists an optimal ratio between $L$ and $M$ in suburban area due to higher LoS probability; whereas, it tends to deploy more AP-UAVs in the urban and dense urban areas.
\end{itemize}

The rest of this article is organized as follows. The channel model, the channel estimation, and the signal model are introduced in Section~\ref{secModel}. In Section~\ref{secDec}, the cell-free detection methods including the fully-centralized detection and the two-layer distributed detection scheme are introduced. In Section~\ref{secAsym}, the asymptotic combining weights of the two-layer distributed detection are calculated by using the operator-valued free probability theory. Numerical results are given in Section~\ref{secResult}. Section~\ref{secConclude} concludes the main findings of this article.

\emph{Notations.} Throughout the paper, vectors are represented by lower-case bold-face letters, and matrices are represented by upper-case bold-face letters. We use $\mb{0}_n$ and $\mb{I}_n$ to represent an $n\times n$ all-zero matrix and an $n\times n$ identity matrix, respectively. The notation ${\rm{diag}}(\mb{A}_i)_{1\leq i\leq n}$ or ${\rm{diag}}(\mb{A}_1,\dots,\mb{A}_n)$ denotes the diagonal block matrix consisting of square matrices $\mb{A}_1,\dots,\mb{A}_n$. The notation $[\mb{A}]_{mn}$ denotes the $(m,n)$-th entry of matrix $\mb A$. The superscript $(\cdot)^{\star}$, $(\cdot)^{\mathrm{T}}$, and $(\cdot)^{\dag}$ are denoted as the conjugate, transpose, and conjugate transpose operations, respectively. We denote $\mathrm{Tr}(\mb{A})$ and ${\rm{vec}}(\mb{A})$ as the trace and vectorization of $n\times n$ matrix $\mb{A}$, respectively. The notation $\mbb{E}[\cdot]$ denotes the expectation. The notation $\otimes$ is denoted as the Kronecker product.

\section{System Model} \label{secModel}
Consider an uplink aerial cell-free communication system consisting of $L$ UAV-mounted APs, one CPU-UAV, and $K$ ground UEs, as illustrated in Fig.~\ref{figModel}. The users' signals are transmitted to the AP-UAVs, and then are wirelessly fronthauled to the CPU-UAV for joint signal detection. Each AP-UAV is equipped with $M$ receiving antennas and each UE is equipped with $N$ transmitting antennas. Denote the total number of receiving and transmitting antennas as $M_{\rm{tot}}=L\times M$ and $N_{\rm{tot}}=N\times K$, respectively. We consider the more practical Rician fading channel, whose model is given in Section~\ref{subsecchanModel}. In addition, we assume that the considered communication system adopts the time-division duplex (TDD) protocol and that the channel responses are reciprocal between AP-UAVs and UEs. The channel responses are assumed to be constant and frequency flat in each coherence block of $\tau_c$-length. Furthermore, each $\tau_c$-length coherence block consists of the $\tau_p$-length channel estimation phase and $\tau_c-\tau_p$-length data transmission phase. The two phases are described below.

\subsection{Channel Model} \label{subsecchanModel}
The ground-to-air channel coefficients between AP-UAV $l$ and UE $k$ are denoted as $\mb{H}_{kl}\in \mbb{C}^{M\times N}$, which follows the Rician MIMO model with Kronecker's correlation structure as
\begin{equation}
	\mb{H}_{kl}=\mb{\ob{H}}_{kl}+\mb{R}_{kl}^{\frac{1}{2}}\mb{W}_{kl}\mb{T}_{kl}^{\frac{1}{2}},
\end{equation}
where $\mb{\ob{H}}_{kl}$ is the deterministic components representing the LoS propagation component and $\mb{W}_{kl}$ denotes the random scattering components consisting of independent identically distributed ($i.i.d.$) Gaussian random variables, i.e., $[\mb{W}_{kl}]_{mn}\in \mc{CN}(0,1)$. The matrices $\mb{R}_{kl}\in \mbb{C}^{M\times M}$ and $\mb{T}_{kl} \in \mbb{C}^{N\times N}$ are the correlation matrices at the receiver and the transmitter sides, respectively. Note that ${\rm{vec}}(\mb{H}_{kl})\in\mc{CN}({\rm{vec}}(\mb{\ob{H}}_{kl}), \mb{C}_{kl})$, where $\mb{C}_{kl}=\mbb{E}[{\rm{vec}}(\mb{H}_{kl}-\mb{\ob{H}}_{kl}){\rm{vec}}(\mb{H}_{kl}-\mb{\ob{H}}_{kl})^{\dag}]=\mb{T}_{kl}^{\rm T} \otimes\mb{R}_{kl}$ is the full correlation matrix~\cite{yu2002models}.

In addition, we denote the distance-dependent pathloss between the AP-UAV $l$ and UE $k$ as $\beta_{kl} = \mbb{E}\{\lVert\mb{H}_{kl}\rVert^2_{\rm{F}}\}/M$ given by
\begin{equation}
	\mbb{E}\{\lVert\mb{H}_{kl}\rVert^2_{\rm{F}}\} = \frac{1}{N}{\rm{Tr}}(\mb{R}_{kl}){\rm{Tr}}(\mb{T}_{kl})+{\rm{Tr}}(\mb{\ob{H}}_{kl}\mb{\ob{H}}_{kl}^\dag)
\end{equation}
Following the standard conventions~\cite{zhang2013capacity}, the channel statistics $\mb{\ob{H}}_{kl}$, $\mb{R}_{kl}$ and $\mb{T}_{kl}$ are normalized such that 
\begin{equation}
	\begin{cases}
		\operatorname{Tr}(\mathbf{R}_{kl})=\frac{1}{\kappa+1}\beta_{kl}M,\\ 
		\operatorname{Tr}(\mathbf{T}_{kl})=N,\\ \operatorname{Tr}\left(\ob{\mathbf{H}}_{kl}\ob{\mathbf{H}}_{kl}^{\dag}\right)=\frac{\kappa}{\kappa+1}\beta_{kl}M,
	\end{cases} 
\end{equation}
where $\kappa$ is the Rician factor.

\subsection{Channel Estimation}
 In the channel estimation phase, $\tau_p$ mutually orthogonal pilot sequences are randomly assigned to UEs. In specific, the UE $k$ is assigned with pilot matrix $\mb{\Phi}_k\in \mbb{C}^{\tau_p\times N}$, where $\mb{\Phi}_k^{\dag}\mb{\Phi}_i=\tau_p\mb{I}_N$, if $k=i$ and $\mb{0}$ otherwise. 

For the uplink communications with massive multi-antenna transmitters, the total number of transmitting antennas can be easily larger than the available number of pilots, i.e., $KN\geq \tau_p$. In this setting, more than one UE has to be assigned with the same pilot matrix. We define the set $\mc{U}_k$ as the subset of UEs, which use the same pilot as UE $k$, and $t_k\in \{1,\ldots, \tau_p\}$ as the index of pilot matrix assigned to UE $k$, respectively. Thus, the received signal $\mb{Y}_l^{tr}\in \mbb{C}^{M\times \tau_p}$ at AP-UAV $l$ in channel estimation phase is given by
\begin{equation}
	\mb{Y}_l^{tr} = \sum_{i=1}^{K} \mb{H}_{il}\mb{F}_i\mb{\Phi}_i^{\rm{T}}+\mb{N}_l^{tr},
\end{equation}
where $\mb{F}_i\in \mbb{C}^{N\times N}$ denotes the precoder of UE $i$ and $\mb{N}_l^{tr}\in \mbb{C}^{M\times \tau_p}$ is the additive noise at AP-UAV $l$ with $i.i.d$ $\mc{CN}(0,\sigma^2)$ entries, where $\sigma^2$ is the noise power. Note that the precoder $\mb{F}_i$ follows the power constrain ${\rm{Tr}}(\mb{F}_i\mb{F}_i^{\dag})\leq p_i$, where $p_i$ is maximum transmitting power of UE $i$.

Then, the AP-UAV $l$ correlates the received signals with the associated pilot $\mb{\Phi}_{k}^\star$ as follow
\begin{equation}
	\begin{aligned}
		\mb{Y}_{kl}^{tr} &= \mb{Y}_l^{tr}\mb{\Phi}_{k}^\star = \sum_{i=1}^{K} \mb{H}_{il}\mb{F}_i\mb{\Phi}_i^{\rm{T}}\mb{\Phi}_{k}^\star+\mb{N}_l^{tr}\mb{\Phi}_{k}^\star, \\
		&= \tau_p \sum_{i\in \mc{U}_k} \mb{H}_{il}\mb{F}_i + \mb{\wt N}_l^{tr},
	\end{aligned}
\end{equation}
where $\mb{\wt N}_l^{tr} = \mb{N}_l^{tr}\mb{\Phi}_{k}^\star$. By vectorizing the received signal $\mb{Y}_{kl}^{tr}$, the received signal can be rewritten as 
\begin{equation}
	\mb{y}_{kl}^{tr} = {\rm{vec}}(\mb{Y}_{kl}^{tr}) =\tau_p \sum_{i\in \mc{U}_k} \mb{\wt F}_i \mb{h}_{il} + {\rm{vec}}(\mb{\wt N}_l^{tr}),
\end{equation}
where $\mb{\wt F}_i = \mb{F}_i^{\rm{T}} \otimes \mb{I}_{M}$ and $\mb{h}_{il}={\rm{vec}}(\mb{H}_{il})$. By using the MMSE estimator~\cite{bjornson2009framework}, the estimated channel $\mb{\wh{H}}_{kl}$ is given by
\begin{equation}
	\begin{aligned}
		\mb{\wh h}_{kl} &= {\rm{vec}}(\mb{\wh H}_{kl}) = {\rm{vec}}(\mb{\ob{H}}_{kl})+{\rm{vec}}(\mb{\wt{H}}_{kl})  \\
		&= \ob{\mb{h}}_{kl}  + \mb{C}_{kl}\mb{\wt F}_k^{\dag}\mb{\Pi}_{kl}^{-1}(\mb{y}_{kl}-\mb{\ob y}_{kl}),
	\end{aligned}
\end{equation}
where $\mb{\ob y}_{kl} = \tau_p \sum_{i\in \mc{U}_k} \mb{\wt F}_i \mb{\ob h}_{il}$ with $\mb{\ob h}_{il} = {\rm{vec}}(\mb{\ob H}_{il})$, $\mb{\wt{H}}_{kl}$ represent the random uncertain components, and $\mb{\Pi}_{kl}$ is the normalized correlation matrix of the received signal $\mb{y}_{kl}^{tr}$, which is given by
\begin{equation}
	\begin{aligned}
		\mb{\Pi}_{kl} &= \frac{1}{\tau_p}\mbb{E}\left\{(\mb{y}_{kl}^{tr}-\mb{\ob y}_{kl})(\mb{y}_{kl}^{tr}-\mb{\ob y}_{kl})^{\dag}\right\}, \\
		&= \tau_p \sum_{i\in \mc{U}_k} \mb{\wt F}_i\mb{C}_{il}\mb{\wt F}_i^{\dag}+\sigma^2\mb{I}_{M_lN}.
	\end{aligned}
\end{equation}
Considering the orthogonality property of MMSE estimation~\cite{bjornson2009framework}, the estimation error $\mb{\breve h}_{kl}={\rm{vec}}(\mb{\breve{H}})=\mb{h}_{kl}-\mb{\wh h}_{kl}\in \mc{CN}(\mb{0},\mb{\wt C}_{kl})$, where $\mb{\wt C}_{kl}=\mb{C}_{kl} - \mb{\wh C}_{kl}$ with $\mb{\wh C}_{kl}=\tau_p\mb{C}_{kl}\mb{\wt F}_k^{\dag}\mb{\Pi}_{kl}^{-1}\mb{\wt F}_k\mb{C}_{kl}$ being the  correlation matrix of the estimation $\mb{\wh h}_{kl}$.

 To describe the correlation of the channel coefficients and facilitate the derivation in Section~\ref{secAsym}, we define the one-sided correlation function of the equivalent channel $\mb{\grave H}_{kl}=\mb{L}\mb{\wt H}_{kl}\mb{R}$ parameterized by a Hermitian matrix $\mb{D}\in\mbb{C}^{N\times N}$ as $\eta_{kl}(\mb{L},\mb{D},\mb{R})=\mbb{E}\left[\mb{\grave H}_{kl}\mb{D}\mb{\grave H}_{kl}^{\dag}\right]$, where $\mb{L}\in\mbb{C}^{M\times M}$ and $\mb{R}\in\mbb{C}^{N\times N}$ are arbitrary matrices. The $(m,n)$-th element of $\eta_{kl}(\mb{L},\mb{D},\mb{R})$ is given by
\begin{equation}
	\begin{aligned}
		\left[\eta_{kl}(\mb{L},\mb{D},\mb{R})\right]_{mn}&=\sum_{i=1}^{N}\sum_{j=1}^{N}\langle\langle\mb{R}\rangle\rangle_j\mb{D}\langle\langle\mb{R}\rangle\rangle_i^{\dag}\langle\langle\mb{L}\rangle\rangle_m\mb{\wh C}_{kl}^{ji}\langle\langle\mb{L}\rangle\rangle^{\dag}_n,
	\end{aligned}
\end{equation}
where the notation $\langle\langle\mb{A}\rangle\rangle_i$ denotes the $i$-th row of the matrix $\mb{A}$. The matrix $\mb{\wh C}_{kl}^{ji}=\mbb{E}[\langle\mb{\wt H}_{kl}\rangle_j\langle\mb{\wt H}_{kl}\rangle_i^{\dag}]$ denotes the $(j,i)$-th sub-matrix of $\mb{\wh C}_{kl}$, with $\langle\mb{\wt H}_{kl}\rangle_i$ and $\langle\mb{\wt H}_{kl}\rangle_j$ being the $i$-th and $j$-th column of $\mb{\wt{H}}_{kl}$, respectively.

The other parameterized one-sided correlation function $\wt{\eta}_{kl}(\mb{L},\mb{\wt D},\mb{R})=\mbb{E}\left[\mb{\grave H}_{kl}^{\dag}\mb{\wt D}\mb{\grave H}_{kl}\right]$ parameterized by $\mb{\wt D}\in\mbb{C}^{M\times M}$, whose $(m,n)$-th element is given by
\begin{equation}
	\begin{aligned}
		\left[\wt{\eta}_{kl}(\mb{L},\mb{\wt D},\mb{R})\right]_{mn}&=\sum_{i=1}^{M}\sum_{j=1}^{M}\langle\mb{L}\rangle_j^{\dag}\mb{\wt D}\langle\mb{L}\rangle_i\langle\mb{R}\rangle_m^{\dag}\mb{\Upsilon}_{kl}^{ji}\langle\mb{R}\rangle_n,
	\end{aligned}
\end{equation}
where the $N\times N$ matrix $\mb{\Upsilon}_{kl}^{ji}=\mbb{E}[\langle\langle\mb{\wt H}_{kl}\rangle\rangle_j^{\dag}\langle\langle\mb{\wt H}_{kl}\rangle\rangle_i]$, whose $(m,n)$-th element is denoted as $\left[\mb{\Upsilon}_{kl}^{ji}\right]_{mn}=\left[\mb{\wh C}_{kl}^{nm}\right]_{ij}$.

\subsection{Uplink Data Transmission}
In the data transmission phase, the received signal $\mb{y}_l\in \mbb{C}^{M}$ at the AP-UAV $l$ is given by
\begin{equation}
	\mb{y}_l = \sum_{k=1}^{K} \mb{H}_{kl}\mb{P}_k \mb{x}_k + \mb{n}_l,
\end{equation}
where $\mb{x}_{k}\in \mbb{C}^{N}$ is the signal transmitted from UE $k$ and $\mb{x}_{k} \in \mc{CN}(0, \mb{I}_{N})$. The precoding matrix $\mb{P}_k$ at the UE $k$ satisfies the power constraint as ${\rm{Tr}}(\mb{P}_k\mb{P}_k^{\dag})\leq p_k$. The additive noise $\mb{n}_l\in \mc{CN}(0,\sigma^2\mb{I}_{M})$ is an $i.i.d.$ complex Gaussian vector.

Based on the Rician fading channel model, in the following sections, we will investigate the decentralized signal detection scheme and derive the closed-form expressions of one-shot combining weights by utilizing the channel statistics.

\section{Decentralized Signal Detection with One-Shot Combining} \label{secDec}
As illustrated in Fig.~\ref{figModel}, several AP-UAVs are connected to a CPU-UAV via the wireless fronthaul, which has limited capacity and uncertain latency. In addition, the computational capability and service range of AP-UAVs are limited by their lightweight fuselage. To achieve joint reception in the aerial CF mMIMO system, a common method called the fully-centralized scheme, where AP-UAVs are treated as remote antennas forwarding the received symbols to the CPU-UAV transparently, is introduced in Section~\ref{FCScheme} and provided as a benchmark. Although the highest SE can be achieved in the fully-centralized scheme, the huge interaction overhead and resulting latency prevent it from being used in practice. Therefore, to take advantage of macro-diversity and reduce the interactive latency, a decentralized reception framework consisting of the local detection at AP-UAVs and the one-shot combining at the CPU-UAV, is proposed in Section~\ref{subsecOneshot}.

\subsection{Fully-centralized Scheme} \label{FCScheme}
In the fully-centralized scheme, the received symbols are centrally detected in the CPU-UAV, which collects channel estimates and the received raw symbols from AP-UAVs. The AP-UAVs act as geographically distributed antennas, which only forward data to the CPU-UAV, and have no computation burden. The fully-centralized scheme can take full advantage of diversity gain and expand coverage to ensure communication quality. 

The collected raw signals in the CPU-UAV can be expressed as 
\begin{equation} \label{fcsignal}
\underbrace{\left[ \begin{array}{c}
		\mathbf{y}_1 \\
		\vdots \\
		\mathbf{y}_L
	\end{array} \right]}_{\triangleq \mathbf{y}}=\underbrace{\left[ \begin{array}{c}
		\mathbf{H}_{1} \\
		\vdots \\
		\mathbf{H}_{L}
	\end{array} \right]}_{{\triangleq \mathbf{H}}}
\underbrace{\left[ \begin{array}{ccc}
		\mathbf{P}_{1} & & \\
		& \ddots &\\
		& & \mathbf{P}_{K}\\
	\end{array} \right]}_{{\triangleq \mathbf{P}}}
\underbrace{\left[ \begin{array}{c}
		\mb{x}_{1} \\
		\vdots \\
		\mb{x}_{K}
	\end{array} \right]}_{{\triangleq \mathbf{x}}}+
\underbrace{\left[ \begin{array}{c}
		\mathbf{n}_1 \\
		\vdots \\
		\mathbf{n}_L
	\end{array} \right]}_{\triangleq \mb n},
\end{equation}
In a more compact form, $\mb{y}$ is rewritten as
\begin{equation}
	\mb{y} = \mb{H}\mb{P}\mb{x}+\mb{n},
\end{equation}
where $\mb{{H}}_l=[\mb{{H}}_{1l},\mb{{H}}_{2l},\dots,\mb{{H}}_{Kl}]\in \mbb{C}^{M\times N_{\rm{tot}}}$ is the aggregative channel of AP-UAV $l$.

The CPU-UAV collects the channel estimates and recovers the received signals with the global MMSE detector, which minimizes the conditional MSE between the original signal $\mb{x}$ and the global estimated signal $\mb{\wh{x}}=\mb{U}^\dag\mb{y}$ given the channel estimation $\wh{\mb{H}}$ as follows
\begin{equation}
	{\rm{M}}_{\rm{FC}} = \underset{\substack{\mb{U}}}
	{\text{min}}~~  \mathbb{E}\{\lVert \mb{\wh{x}} - \mb x \rVert ^2|\mb{\wh H}\},
\end{equation}
where $\mb{\wh H}\in\mbb{C}^{M_{\rm{tot}}\times N_{\rm{tot}}}$ is the aggregative channel estimates, and has the same structure as $\mb{H}$ in (\ref{fcsignal}). The global MMSE detector $\mb{U}\in \mbb{C}^{M_{\rm{tot}}\times N_{\rm{tot}}}$ is given by
\begin{equation} \label{globalU}
	\mb{U} = \left(\mb{\wh H}\mb{P}\mb{P}^{\dag}\mb{\wh H}^\dag+\mb{\wt{C}}^{\prime}+\sigma^2\mb{I}_{M_{\rm{tot}}}\right)^{-1}\mb{\wh H}\mb{P},
\end{equation}
where $\mb{\wt{C}}^{\prime}={\rm{diag}}(\mb{\wt{C}}_{l}^{\prime})_{1\leq l\leq L} \in \mbb{C}^{M_{\rm{tot}}\times M_{\rm{tot}}}$. The Hermitian matrix $\mb{\wt{C}}_{l}^{\prime} = \sum_{k=1}^{K}\mb{\wt{C}}_{kl}^{\prime} \in \mbb{C}^{M\times M}$, with $\mb{\wt{C}}_{kl}^{\prime} =\mbb{E}(\mb{\breve H}_{kl}\mb{P}_k\mb{P}^{\dag}_k\mb{\breve H}_{kl}^{\dag})$, whose $(m,n)$-th element is given by
\begin{equation}
	[\mb{\wt{C}}_{kl}^{\prime}]_{mn}=\sum_{n_1=1}^{N}\sum_{n_2=1}^{N}[\mb{P}_k\mb{P}_k^{\dag}]_{n_2n_1}[\mb{\wt C}_{kl}^{n_2n_1}]_{mn},
\end{equation}
where $\mb{\wt C}_{kl}^{n_2n_1}=\mbb{E}[\langle\mb{\breve{H}}_{kl}\rangle_{n_2}\langle\mb{\breve{H}}_{kl}\rangle_{n_1}^{\dag}]$ is the $(n_2,n_1)$-th sub-matrix of $\mb{\wt C}_{kl}$.

Based on~(\ref{globalU}), the detected signals of UE $k$ can be written as
\begin{equation} \label{geqx}
	\begin{aligned}
		\mb{\wh{x}}_k &= \mb{U}_k^{\dag}\mb{y}
		= \mb{U}_k^{\dag}\mb{H}_k\mb{P}_k\mb{x}_k+\sum_{i=1\backslash k}^{K}\mb{U}_k^{\dag}\mb{H}_i\mb{P}_i\mb{x}_i + \mb{U}_k^{\dag}\mb{n},
	\end{aligned}
\end{equation}
where $\mb{U}_k = \left(\mb{\wh H}\mb{P}\mb{P}^{\dag}\mb{\wh H}^\dag+\mb{\wt{C}}^{\prime}+\sigma^2\mb{I}_{M_{\rm{tot}}}\right)^{-1}\mb{\wh H}_k\mb{P}_k$. The channel coefficient $\mb{H}_k=[\mb{H}_{k1}^{\dag},\ldots,\mb{H}_{kL}^{\dag}]^{\dag}\in\mbb{C}^{M_{\rm{tot}}\times N}$ and the channel estimates $\mb{\wh H}_k$ have the same structure as $\mb{H}_k$. 
Based on~(\ref{geqx}), the achievable SE of UE $k$ using per-user-basis MMSE-SIC detector~\cite{li2016massive} is given by
\begin{equation} \label{SEFC}
	{\rm{SE}}_k^{\rm{FC}}=(1-\frac{\tau_p}{\tau_c})\mbb{E}\left\{\log_2\lvert\mb{I}_N+{\rm{SINR}}_k^{\rm{FC}}\rvert\right\},
\end{equation}
where
\begin{equation}
	{\rm{SINR}}_k^{\rm{FC}}=\frac{\mb{P}_k^{\dag}\mb{\wh H}_k^{\dag}\mb{U}_k\mb{U}_k^{\dag}\mb{\wh H}_k\mb{P}_k}{\mb{U}_k^{\dag}\left(\sum_{i=1\backslash k}^{K}\mb{\wh H}_i\mb{P}_i\mb{P}_i^{\dag}\mb{\wh H}_i^{\dag}+\mb{\wt{C}}^{\prime}+\sigma^2\mb{I}_M\right)\mb{U}_k}.
\end{equation}

The fully-centralized scheme with the global MMSE detector not only minimizes the MSE between the global estimated signals $\mb{\wh x}$ and the original signals $\mb{x}$, but also maximizes the SE~\cite[Th.~1]{wang2022uplink}. However, the aggregation of instantaneous channel coefficients and raw signals at the CPU-UAV requires large transmission bandwidth, which incurs unpredictable latency. In order to overcome this drawback, we propose the distributed MMSE detection with one-shot combining scheme (MMSE-OneShot), which is a two-layer signal detection scheme, including the local MMSE detection in AP-UAVs and one-shot combining in the CPU-UAV.

\subsection{Distributed MMSE Detection with One-Shot Combining} \label{subsecOneshot}
In this subsection, we will introduce the two-layer detection scheme. The main difference between the fully-centralized scheme and the MMSE-OneShot scheme is the interaction of information. In the fully-centralized scheme, the instantaneous channel coefficients and raw symbols require to be forwarded from AP-UAVs to the CPU-UAV in every coherence block. Nevertheless, in the MMSE-OneShot scheme, the CPU-UAV requires the local detected signals from each AP-UAV, as well as the long-term channel statistics between AP-UAVs and UEs, to obtain the jointly detected signals. Therefore, only the vector-sized signals need to be forwarded to the CPU-UAV in every coherence block, while the channel statistics, changing much slower than the instantaneous channel coefficients\footnote{The channel statistics are independent of frequency and about 40 times slower than the instantaneous channel~\cite{nayebi2016performance}} only need to be forwarded in every several consecutive blocks. The details of the proposed MMSE-OneShot scheme are described as follows.

In the first layer of the MMSE-OneShot scheme, AP-UAVs individually recover the received signals with the local MMSE detector based on the local instantaneous channel coefficients. The detected signal at AP-UAV $l$ is given by
\begin{equation}
	\mb{\wh x}_l = \mb{U}_l^{\dag}\mb{y}_l,
\end{equation}
where the local MMSE detector $\mb{U}_l\in \mbb{C}^{M\times N_{\rm{tot}}}$ minimizes the MSE between the original signals $\mb{x}$ and the detected signals $\mb{\wh{x}}_l$, which is given by
\begin{equation}
	\mb{U}_l = (\mb{\wh{H}}_l\mb{P}\mb{P}^{\dag}\mb{\wh H}_l^{\dag}+\mb{\wt{C}}_{l}^{\prime}+\sigma^2\mb{I}_{M})^{-1}\mb{\wh{H}}_l\mb{P},
\end{equation}
where $\mb{\wh{H}}_l=[\mb{\wh{H}}_{1l},\mb{\wh{H}}_{2l},\dots,\mb{\wh{H}}_{Kl}]\in \mbb{C}^{M\times N_{\rm{tot}}}$ is the aggregative channel estimates at AP-UAV $l$. Note that the detector $\mb{U}_l$ only depends on the local channel estimates, thus there is no inter-UAV interaction required in the first layer.

In the second layer, after detecting the received signal using the local detector $\{\mb{U}_l\}_{1\leq l\leq L}$, AP-UAVs forward the detected signals $\{\mb{\wh x}_l\}_{1\leq l\leq L}$ to the CPU-UAV. Then, the CPU-UAV collects and weighted combines the detected signals as follows
\begin{equation} \label{cbsignals}
	\mb{\wh x}_{\mathrm{dist}} = \sum_{l=1}^{L}\omega_l\mb{\wh x}_l,
\end{equation}
where $\bm{\omega}=[\omega_1,\dots,\omega_L]^{\rm{T}}\in\mbb{C}^{L}$ is the assigned combining weights and $\mb{\wh x}_{\mathrm{dist}}\in\mbb{C}^{N_{\rm{tot}}}$ is the final refined signal in the CPU-UAV. The optimal combining weights $\bm{\omega}$ can be obtained by minimizing the MSE between the original signal $\mb{x}$ and the refined signal $\mb{\wh x}_{\mathrm{dist}}$, which is given as follows

\begin{equation}
	{\rm{M}}_{\rm{dist}} = \underset{\substack{\bm{\omega}}}
	{\text{min}}~~  \mathbb{E}\lVert \mb{\wh{x}}_{\mathrm{dist}} - \mb x \rVert ^2,
\end{equation}
where the expectation is taken with respect to the unknown noise. The closed-form expressions of the optimal combining weights and the corresponding MSE are given in the following proposition.
 
\begin{proposition}
	The optimal combining weights and the corresponding MSE are given by
	\begin{align}
		\bm{\omega} &= (\mb{A}+\mb{Y})^{-1}\mb{v}, \label{weights} \\
		{\rm{M}}_{\rm{dist}} &= \lVert \mb x \rVert^2 - \mb v^\dag (\mb A+\mb Y)^{-1}\mb{v}, \label{MSE}
	\end{align}
	where $\mb{v}\in \mbb{C}^{L}$, $\mb{A}\in \mbb{C}^{L\times L}$ and $\mb{Y}\in \mbb{C}^{L\times L}$, whose entries are given by
	\begin{align}
		[\mb{v}]_l &=\mb{x}^{\dag}\mb{Q}_l\mb{x},\\
		[\mb{A}]_{lm} &= \mb{x}^{\dag}\mb{Q}_l\mb{Q}_m\mb{x},\\
		[\mb{Y}]_{ll} &= {\rm{Tr}}(\mb{U}_l^{\dag}(\mb{\wt{C}}_{l}^{\prime}+\sigma^2\mb{I}_{M}) \mb{U}_l),
	\end{align}
	and
	\begin{equation}
		\mb{Q}_l= \mb{P}^{\dag}\mb{\wh{H}}_l^{\dag}(\mb{\wh{H}}_l\mb{P}\mb{P}^{\dag}\mb{\wh H}_l^{\dag}+\mb{\wt{C}}_{l}^{\prime}+\sigma^2\mb{I}_{M})^{-1}\mb{\wh{H}}_l\mb{P},
	\end{equation}
\end{proposition}
\begin{proof}
	The proof follows similar steps as the proof of \cite[Prop. 1]{pan2023design}, and is therefore omitted.
\end{proof}

As described in the above proposition, the calculation of combining weights requires the unknown original signal $\mb{x}$, which is unrealistic in the implementation. In addition, the calculation of $\mb{Q}_l$ and $\mb{U}_l$ in the CPU-UAV still needs the instantaneous channel estimates, which are unavailable in the CPU-UAV and need to be transmitted from distributed AP-UAVs. Instead, when the numbers of receiving antennas per AP-UAV $M$ and transmitting antennas per UE $N$ grow to infinity with a fixed ratio, we can utilize the free probability theory to obtain the asymptotic expressions of the combining weights, which only depend on the long-term statistical CSI. In this way, we can still exploit the diversity gain due to multi-UAV joint detection, while expanding the range of wireless coverage. More practically, the amount of information interaction between the CPU-UAV and AP-UAVs is much reduced, compared to the fully centralized scheme.

We first approximate the expressions of $\mb{v}$, $\mb{A}$, and $\mb{Y}$ in~(\ref{weights}) in terms of the Cauchy transform using the large dimensional random matrix theory. Then, we derive the explicit expressions of the Cauchy transform to obtain the asymptotic expressions of the combining weights via operator-valued free probability theory, which will be introduced in Section~\ref{secAsym}.

According to the trace lemma~\cite[Th.~3.4]{couillet2011random}, the entries in $\mb{v}$ and $\mb{A}$ can be rewritten as
 \begin{align}
 	[\mb{v}]_l &=\mb{x}^{\dag}\mb{Q}_l\mb{x}\to_{a.s.} {\rm{Tr}}(\mb{Q}_l),\\
 	[\mb{A}]_{lm} &= \mb{x}^{\dag}\mb{Q}_l\mb{Q}_m\mb{x}\to_{a.s.} {\rm{Tr}}(\mb{Q}_l\mb{Q}_m),
 \end{align}
where $a.s.$ denotes almost surely, which represents that the left-hand-side (LHS) will converge to right-hand-side (RHS) almost surely as $N_{\rm{tot}}\to \infty$.
Based on the matrix inversion lemma~\cite{couillet2011random}, the matrix $\mb{Q}_l$ can be rewritten as 
\begin{equation}
	\mb{Q}_l = \mb{I}_{N_{\rm{tot}}} - \left(\mb{P}^{\dag}\mb{\wh{H}}_l^{\dag}(\mb{\wt{C}}_{l}^{\prime}+\sigma^2\mb{I}_{M})^{-1} \mb{\wh{H}}_l\mb{P}+\mb{I}_{N_{\rm{tot}}}\right)^{-1}.
\end{equation} 

In addition, we define the Cauchy transform of $\mb{B}_l=\mb{P}^{\dag}\mb{\wh{H}}_l^{\dag}(\mb{\wt{C}}_{l}^{\prime}+\sigma^2\mb{I}_{M})^{-1} \mb{\wh{H}}_l\mb{P}\in\mbb{C}^{N_{\rm{tot}}\times N_{\rm{tot}}}$ as 
\begin{equation} \label{Cauchy1}
	\mc{G}_{\mb{B}_l}(z) =  \frac{1}{N_{\rm{tot}}}\mathrm{Tr}\circ \mbb{E}\left\{\left(z\mb{I}_{N_{\rm{tot}}} - \mb{B}_l\right)^{-1}\right\}.
\end{equation}
Then, we can rewrite the expressions of $\mb{v}$ and $\mb{A}$ in terms of $\mc{G}_{\mb{B}_l}(z)$. The $l$-th element in $\mb v$ can be then expressed as
 \begin{equation}
 	\mb{v}_l= N_{\rm{tot}}+N_{\rm{tot}}\mc{G}_{\mb{B}_l}(z),
 \end{equation}

Similarly, the $(l,m)$-th entry in $\mb{A}$ is given by
\begin{equation} \label{matrixA}
	\begin{aligned}
		[\mb{A}]_{lm} = 
		 \begin{cases}
			N_{\rm{tot}}+N_{\rm{tot}}\mc{G}_{\mb{B}_l}(z)+N_{\rm{tot}}\mc{G}_{\mb{B}_m}(z)+N_{\rm{tot}}^2{\rm{Tr}}(\mb{B}_l\mb{B}_m), \qquad
			 \text{$l \neq m$},\\ 
			N_{\rm{tot}}+2N_{\rm{tot}}\mc{G}_{\mb{B}_l}(z)+ N_{\rm{tot}}\mc{G}^\prime_{\mb{B}_l}(z),
			\qquad\qquad\qquad\qquad\;\;\;\text{$l = m$},
		\end{cases}
	\end{aligned}
\end{equation}
where $z=-1$ in the above two formulas and $\mc{G}^\prime_{\mb{B}_l}(z)=\frac{1}{N_{\rm{tot}}}\mbb{E}\{{\rm{Tr}}(z\mb{I}_{N_{\rm{tot}}}-\mb{B}_l)^{-2}\}$ is the derivative of the Cauchy transform $\mc{G}_{\mb{B}_l}(z)$. The last item of $[\mb{A}]_{lm},l\neq m$ can be obtained by deriving the approximation of $\mb{B}_{l},l=1,\dots,L$, which will be detailed in Section~\ref{secAsym}. 
 
 Besides, the $l$-th diagonal element in $\mb{Y}$ can be also denoted as
 \begin{equation} \label{Y}
 	\begin{aligned}
 		[\mb{Y}]_{ll}&={\rm{Tr}}((\mb{\wh H}_l\mb{P}\mb{P}^{\dag}\mb{\wh H}_l^{\dag}+\mb{S}_l)^{-1}\mb{S}_l)
 		-{\rm{Tr}}((\mb{\wh H}_l\mb{P}\mb{P}^{\dag}\mb{\wh H}_l^{\dag}+\mb{S}_l)^{-2}\mb{S}_l\mb{S}_l^{\dag}),
 	\end{aligned}
 \end{equation}
 where $\mb{S}_l = \mb{\wt{C}}_{l}^{\prime}+\sigma^2\mb{I}_{M}$.
 
 Define the Cauchy transform of $\mb{\wt{B}}_l=\mb{\wh H}_l\mb{P}\mb{P}^{\dag}\mb{\wh H}_l^{\dag}+\mb{\wt{C}}_{l}^{\prime} \in\mbb{C}^{M\times M}$ as 
 \begin{equation} \label{Cauchy2}
 	\mc{G}_{\mb{\wt{B}}_l,\mb{\Xi}}=\frac{1}{M}\mathrm{Tr}\circ\mbb{E}\left\{(z\mb{I}_{M}-\mb{\wt B}_l)^{-1}{\mb{\Xi}}\right\},
 \end{equation} 
 where $\mb{\Xi}\in \mbb{C}^{M\times M}$ is a nonnegative definite matrix with the uniformly bounded spectral norm. Then, the $l$-th diagonal element in $\mb{Y}$ can be rewritten as
 \begin{equation} \label{matrixY}
 	[\mb{Y}]_{ll}=-M\mc{G}_{\mb{\wt{B}}_l,\mb{S}_l}(z)-M\mc{G}^\prime_{\mb{\wt{B}}_l, \mb{S}_l\mb{S}_l^{\dag}}(z),
 \end{equation}
 where $z=-\sigma^2$ and $\mc{G}^\prime_{\mb{\wt{B}}_l, \mb{S}_l\mb{S}_l^{\dag}}(z)=\frac{1}{M}\mbb{E}\{{\rm{Tr}}((z\mb{I}_{M}-\mb{\wt B}_l)^{-2}\mb{S}_l\mb{S}_l^{\dag})\}$ is the derivative of the Cauchy transform $\mc{G}_{\mb{\wt{B}}_l, \mb{S}_l\mb{S}_l^{\dag}}(z)$.
Therefore, the approximation of the combining weights $\bm{\omega}$ is amount to finding the Cauchy transform $\mc{G}_{\mb{B}_l}(z)$ and $\mc{G}_{\mb{\wt B}_l}(z)$.

Based on the proposed MMSE-OneShot scheme, the detected signals of UE $k$ in~(\ref{cbsignals}) can be rewritten as 
\begin{equation}
		\mb{\wh x}_{k}= \sum_{l=1}^{L}\omega_l\mb{\wh x}_{kl}
		= \sum_{l=1}^{L}\omega_l\left(\mb{U}_{kl}^{\dag}\mb{H}_{kl}\mb{P}_k \mb{x}_k +  \sum_{i=1\backslash k}^{K} \mb{U}_{kl}^{\dag}\mb{H}_{il}\mb{P}_i \mb{x}_i + \mb{U}_{kl}^{\dag}\mb{n}_l\right),
\end{equation}
where $\mb{U}_{kl} = (\mb{\wh{H}}_l\mb{P}\mb{P}^{\dag}\mb{\wh H}_l^{\dag}+\mb{\wt{C}}_{l}^{\prime}+\sigma^2\mb{I}_{M})^{-1}\mb{\wh{H}}_{kl}\mb{P}_k \in \mbb{C}^{M\times N}$. Since the CPU-UAV is only aware of the statistical CSI, the achievable SE using MMSE-SIC detector is given in the following corollary~\cite{wang2022uplink}.
\begin{corollary} \label{coro1}
	The achievable SE of UE $k$ with MMSE-SIC detector is given by
	\begin{equation}
		{\rm{SE}}_k = (1-\frac{\tau_p}{\tau_c})\log_2\lvert\mb{I}_N+\mb{S}_k^{\dag}\mb{\Gamma}_k^{-1}\mb{S}_k\rvert,
	\end{equation}
	where $\mb{S}_k=\sum_{l=1}^{L}\omega_l\mbb{E}\left(\mb{U}^{\dag}_{kl}\mb{H}_{kl}\right)\mb{P}_k$ and $\mb{\Gamma}_k=\sum_{i=1}^{K}\sum_{l=1}^{L}\sum_{m=1}^{L}\omega_l\mbb{E}\left(\mb{U}_{kl}^{\dag}\mb{H}_{il}\mb{P}_i\mb{P}_i^{\dag}\mb{H}_{im}^{\dag}\mb{U}_{km}\right)\omega_m^\star-\mb{S}_k\mb{S}_k^{\dag}+\sigma^2\sum_{l=1}^{L}\omega_l\mbb{E}(\mb{U}_{kl}^{\dag}\mb{U}_{kl})\omega_l^\star$.
\end{corollary}
\begin{proof}
	The proof follows similar steps as the proof of \cite[Th. 2]{li2016massive}, and is therefore omitted.
\end{proof}
In this section, we introduce the fully-centralized scheme as a benchmark and propose the MMSE-OneShot scheme with the optimal closed-form expressions of the combining weights. To obtain the theoretical expressions, we then represent the optimal one-shot combining weights in terms of Cauchy transforms $\mc{G}_{{\mb B}_l}(z)$ and $\mc{G}_{\mb{\wt B}_l}(z)$ using the large dimensional random matrix theory, whose explicit expressions will be detailed in the following section.

\section{Asymptotic One-Shot Combining Weights via Operator-Valued Free Probability Theory} \label{secAsym}
 In this section, we will resort to the linearization trick and the operator-valued free probability theory to obtain the explicit expressions of $\mc{G}_{{\mb B}_l}(z)$ and $\mc{G}_{\mb{\wt B}_l}(z)$. In the following discussion, we will focus on the $l$-th AP-UAV and thus omit the subscript $l$ of $\mb{B}_l$ and $\mb{\wt B}_l$ for brevity.

In the considered formulations, where $\mb{B}=\mb{P}^{\dag}\mb{\wh{H}}_l^{\dag}(\mb{\wt{C}}_{l}^{\prime}+\sigma^2\mb{I}_{M})^{-1} \mb{\wh{H}}_l\mb{P}$ and $\mb{\wt B} = \mb{\wh H}_l\mb{P}\mb{P}^{\dag}\mb{\wh H}_l^{\dag}+\mb{\wt{C}}_{l}^{\prime}$, the equivalent channels in $\mb{B}$ and $\mb{\wt B}$ (which will be defined in the following subsections) are non-central and with non-trivial spatial correlations, and thus not free over the complex algebra~\cite{lu2016free}. 
To address this issue, the Anderson’s linearization trick is adopted~\cite{belinschi2017analytic}. The linearization trick transforms a complicated distribution of an arbitrary self-adjoint polynomial, such as $p\in \mbb{C}(\mb{X}_1,\dots,\mb{X}_n)\in \mc{A}$ with $\mc{A}$ being a complex and unital algebra, into the operator-valued distribution of a linear polynomial $\mb{L}_p=b_0\otimes 1+b_1\otimes \mb{X}_1+\dots+b_n\otimes \mb{X}_n$, with $b_0,\dots,b_n\in \mb{M}_{N}(\mbb{C})$, where $\mb{M}_{N}(\mbb{C})$ is the algebra of complex $N\times N$ matrices. For any complex number $z\in \mbb{C}$, the following conditions are equivalent
\begin{itemize}
	\item The operator $z\mb{I}_N-\mb{P}$ with $\mb{P}:=p(\mb{X}_1,\dots,\mb{X}_n)$ is invertible in $\mc{A}$.
	\item The operator $\mb{\Lambda}(z)-\mb{L}_\mb{P}$ is invertible in $\mb{M}_{N}(\mbb{C})\otimes \mc{A}$, where
	\begin{equation}
		\mb{\Lambda}(z)=
		\left[\begin{array}{c:c}
			z\mb{I}_N & \mb{0}\\ \hdashline
			\mb{0} & \mb{0} \\
		\end{array}\right].
	\end{equation}
\end{itemize}
Then, for some $z\in \mbb{C}$ that fulfills the above conditions, we have~\cite{belinschi2017analytic}
\begin{equation} \label{trick}
	\left\{(\mb{\Lambda}(z)-\mb{L}_\mb{P})^{-1}\right\}^{(1,1)}=(z-\mb{P})^{-1}.
\end{equation}
where $\{\cdot\}^{(1,1)}$ denotes the upper-left $N\times N$ matrix block.

 Therefore, by using the linearization trick, we can first rewrite $\mc{G}_{{\mb B}}(z)$ and $\mc{G}_{\mb{\wt B}}(z)$ as their operator-valued Cauchy transform counterparts, whose expressions can be derived by the operator-valued free probability theory~\cite{voiculescu1995operations}. Then we can obtain the expressions of Cauchy transforms $\mc{G}_{{\mb B}}(z)$ and $\mc{G}_{\mb{\wt B}}(z)$ based on~(\ref{trick}).

\subsection{Cauchy Transform $\mc{G}_{{\mb B}}(z)$}
Recall the Cauchy transform $\mc{G}_{\mb{B}}(z)=\frac{1}{N_{\rm{tot}}}\mbb{E}\{{\rm{Tr}}(z\mb{I}_{N_{\rm{tot}}}-\mb{B})^{-1}\}$ with $\mb{B}=\mb{P}^{\dag}\mb{\wh{H}}_l^{\dag}(\mb{\wt{C}}_{l}^{\prime}+\sigma^2\mb{I}_{M})^{-1} \mb{\wh{H}}_l\mb{P}$, we define the equivalent channel as $\mb{G}=(\mb{\wt C}_{l}^\prime+\sigma^2\mb{I}_{M})^{-\frac{1}{2}}\mb{\wh H}_l\mb{P}\in \mbb{C}^{M\times N_{\rm{tot}}}$ and then we have $\mb{B}=\mb{G}^{\dag}\mb{G}$.

By using the Anderson’s linearization trick, we can construct the block matrix $\mb{L}_\mb{B}\in \mc{M}_n$, where $\mc{M}_n = \mb{M}_n(\mbb{C})$ denotes the algebra of $n\times n$ complex random matrices with $n = N_{\rm{tot}}+M$, as follows
\begin{align}
	\mb{L}_{\mb{B}}=
	\begin{bmatrix}
		\mb{0}_{N_{\rm{tot}}} & \mb{G}^{\dag} \\
		\mb{G} & -\mb{I}_{M}
	\end{bmatrix},
\end{align}

We define the sub-algebra $\mc{D}_n\subset \mc{M}_n$ as the $n\times n$ block diagonal matrix. For each $\mb{K}_n\in \mc{D}_n$, we have
\begin{align} \label{K1}
	\mb{K}_n=
	\begin{bmatrix}
		\mb{D}_1 &  \mb{0}_{N_{\rm{tot}}\times M} \\
		\mb{0}_{M \times N_{\rm{tot}}} & \mb{D}_2
	\end{bmatrix},
\end{align}
where the $N_{\rm{tot}}\times N_{\rm{tot}}$ block diagonal matrix $\mb{D}_1={\rm{blkdiag}}(\mb{D}_{11},\ldots, \mb{D}_{1K})$ with $\{\mb{D}_{1k}\}_{1\leq k\leq K}$ being $N\times N$ sub-matrices and $\mb{D}_2$ is a $M\times M$ sub-matrix.

The $\mc{D}_n$-valued Cauchy transform is then defined as
\begin{align} \label{GL}
	\mc{G}_{\mb{L}_{\mb{B}}}^{\mc{D}_n}(\mb{\Lambda}_n(z)) = \mbb{E}_{\mc{D}_n}\left[(\mb{\Lambda}_n(z)-\mb{L}_{\mb{B}})^{-1}\right],
\end{align}
where the expectation $\mbb{E}_{\mc{D}_n}[\mb{X}]$ has the same structure as $\mb{K}_n$ and is a linear function of $\mb{X}\in \mc{M}_n$, which is defined as
\begin{equation}
	\mbb{E}_{\mc{D}_n}[\mb{X}] = 
	\begin{bmatrix}
		\mbb{E}[\mb{X}_{\mb{D}_1}] &  \mb{0}_{N_{\rm{tot}}\times M} \\
		\mb{0}_{M \times N_{\rm{tot}}} & \mbb{E}[\mb{X}_{\mb{D}_2}]
	\end{bmatrix},
\end{equation}
and the expectation $\mbb{E}[\mb{X}_{\mb{D}_1}]={\rm{blkdiag}}(\mbb{E}[\mb{X}_{\mb{D}_{11}}],\ldots, \mbb{E}[\mb{X}_{\mb{D}_{1K}}])$. In addition, the $n\times n$ diagonal matrix $\mb{\Lambda}_n(z)$ in (\ref{GL}) is given by
\begin{align}
	\mb{\Lambda}_n(z) = \begin{bmatrix}
		z\mb{I}_{N_{\rm{tot}}} & \mb{0}_{N_{\rm{tot}}\times M}\\
		\mb{0}_{M\times N_{\rm{tot}}} & \mb{0}_{M} 
	\end{bmatrix}.
\end{align}

Based on the block matrix inversion lemma, we have
\begin{align}
	\mc{G}_{\mb{L}_{\mb{B}}}^{\mc{D}_n}(\mb{\Lambda}_n(z)) = \mbb{E}_{\mc{D}_n}
	\begin{bmatrix}
		(z\mb{I}_{N_{\rm{tot}}}-\mb{G}^{\dag}\mb{G})^{-1} &
		\mb{G}^{\dag}(z\mb{I}_{M}-\mb{G}\mb{G}^{\dag})^{-1} \\
		(z\mb{I}_{M}-\mb{G}\mb{G}^{\dag})^{-1}\mb{G} &
		z(z\mb{I}_{M}-\mb{G}\mb{G}^{\dag})^{-1}
	\end{bmatrix} ,
\end{align}

Obviously, the upper-left $N_{\rm{tot}}\times N_{\rm{tot}}$ matrix block of $\mc{G}_{\mb{L}_{\mb{B}}}^{\mc{D}_n}(\mb{\Lambda}_n(z))$ is equivalent to $(z\mb{I}_{N_{\rm{tot}}} - \mb{B})^{-1}$, thus the Cauchy transform of $\mb{B}$ is amount to 
\begin{align} \label{iden1}
	\mc{G}_{\mb{B}}(z) = \frac{1}{N_{\rm{tot}}}{\rm{Tr}}\left(\left\{\mc{G}_{\mb{L}_{\mb{B}}}^{\mc{D}_n}(\mb{\Lambda}_n(z))\right\}^{(1,1)}\right),
\end{align}
where $\{\cdot\}^{(1,1)}$ denotes the upper-left $N_{\rm{tot}}\times N_{\rm{tot}}$ matrix block. Henceforth, the Cauchy transform $\mc{G}_{\mb{B}}$ can be obtained by the operator-valued Cauchy transform. Then, in order to obtain the expression of $\mc{G}_{\mb{L}_{\mb{B}}}^{\mc{D}_n}(\mb{\Lambda}_n(z))$, we first separate $\mb{L}_\mb{B}$ into $\mb{\ob{L}}_{\mb{B}}$ and $\mb{\wt L}_\mb{B}$, which are defined as
\begin{align}
	\mb{\ob L}_{\mb{B}} =
	\begin{bmatrix}
		\mb{0}_{N_{\rm{tot}}} & \mb{\ob G}^{\dag} \\
		\mb{\ob G} & -\mb{I}_{M}
	\end{bmatrix}, \qquad
	\mb{\wt L}_{\mb{B}} =
	\begin{bmatrix}
		\mb{0}_{N_{\rm{tot}}} & \mb{\wt G}^{\dag} \\
		\mb{\wt G} & \mb{0}_{M}
	\end{bmatrix},
\end{align}
where $\mb{L}_\mb{B} =\mb{\ob{L}}_{\mb{B}}+\mb{\wt L}_\mb{B}$, $\mb{\ob G} =(\mb{\wt C}_{l}^\prime+\sigma^2\mb{I}_{M})^{-\frac{1}{2}}\mb{\ob H}_l\mb{P}$, and $\mb{\wt G} =(\mb{\wt C}_{l}^\prime+\sigma^2\mb{I}_{M})^{-\frac{1}{2}}\mb{\wt H}_l\mb{P}$.

\begin{proposition} \label{prop2}
	The random variable $\mb{\wt L}_{\mb{B}}$ is semicircular and free from $\mb{\ob L}_{\mb{B}}$ over $\mc{D}_n$.
\end{proposition}
\begin{proof}
	The proof of Proposition~\ref{prop2} is given in Appendix~\ref{Proof2}.
\end{proof}

Based on Proposition~\ref{prop2}, the limiting spectral distribution of $\mb{L}_\mb{B}$ can be determined by the operator-valued free additive convolution of $\mb{\ob{L}}_{\mb{B}}$ and $\mb{\wt L}_\mb{B}$. Thus, the operator-valued Cauchy transform $\mc{G}_{\mb{L}_{\mb{B}}}^{\mc{D}_n}(\mb{\Lambda}_n(z))$ can be obtained by using the subordination formula~\cite{belinschi2017analytic}, which is given as follows
\begin{equation} \label{sub1}
	\begin{aligned}
		\mc{G}_{\mb{L}_{\mb{B}}}^{\mc{D}_n}(\mb{\Lambda}_n(z)) &= \mc{G}_{\ob{\mb{L}}_{\mb{B}}}^{\mc{D}}\left(\mb{\Lambda}_n(z) - \mc{R}_{\widetilde{\mb{L}}_{\mb{B}}}^{\mc{D}_n}\left(\mc{G}_{\mb{L}_{\mb{B}}}^{\mc{D}_n}(\mb{\Lambda}_n(z))\right)\right), \\
		&= \mbb{E}_{\mc{D}_n}\left[\left(\mb{\Lambda}_n(z) - \mc{R}_{\widetilde{\mb{L}}_{\mb{B}}}^{\mc{D}_n}\left(\mc{G}_{\mb{L}_{\mb{B}}}^{\mc{D}_n}(\mb{\Lambda}_n(z))\right) - \ob{\mb{L}}_{\mb{B}}\right)^{-1}\right],
	\end{aligned}
\end{equation}
where $\mc{R}_{\widetilde{\mb{L}}_{\mb{B}}}^{\mc{D}_n}(\cdot)$ denotes the $\mc{D}_n$-valued $R$-transform of ${\mb{L}}_{\mb{B}}$. Therefore, based on the identity~(\ref{iden1}), the expression of the Cauchy transform $\mc{G}_{\mb{B}}(z)$ can be obtain through the expression of $\mc{G}_{\mb{L}_{\mb{B}}}^{\mc{D}_n}(\mb{\Lambda}_n(z))$, which is given by the following proposition.
\begin{proposition} \label{prop3}
	The expression of Cauchy transform $\mc{G}_{\mb{B}}(z)$ is given by
	\begin{equation} \label{equGB}
		\mc{G}_{\mb{B}}(z) = \frac{1}{N_{\rm{tot}}}{\rm{Tr}}\left(\mc{G}_{\mb{D}_1}(z)\right),
	\end{equation}
where $\mc{G}_{\mb{D}_1}(z)$ satisfies the following matrix-valued fix-point equations
\begin{align}
	\mc{G}_{\mb{D}_1}(z) &= \left(\mb{\Psi}(z)-\mb{\ob G}^{\dag}\wt{\mb{\Psi}}^{-1}(z)\mb{\ob G}\right)^{-1},\\
	\mc{G}_{\mb{D}_2}(z) &= \left(\mb{\wt \Psi}(z)-\mb{\ob G}{\mb{\Psi}}^{-1}(z)\mb{\ob G}^{\dag}\right)^{-1},
\end{align}
The matrix-valued function $\mb{\Psi}(z)$ and $\mb{\wt{\Psi}}(z)$ are respectively denoted as
\begin{align}
	\mb{\Psi}(z)&=z\mb{I}_{N_{\rm{tot}}}-\wt{\delta}(\mc{G}_{\mb{D}_2}(z)), \\
	\mb{\wt{\Psi}}(z)&=\mb{I}_{M}-\sum_{i=1}^{K}\delta_i(\mc{G}_{\mb{D}_{1i}}(z)).
\end{align}
The functions of one-sided correlation matrices are given by
\begin{align}
	{\delta}_i(\mc{G}_{\mb{D}_{1i}}(z)) &={\eta}_{il}\left((\mb{\wt C}_{l}^\prime+\sigma^2\mb{I}_{M})^{-\frac{1}{2}}, \mc{G}_{\mb{D}_{1i}}(z),\mb{P}_i\right),
	\\
	\wt{\delta}(\mc{G}_{\mb{D}_2}(z))&={\rm{blkdiag}}\left(\wt{\eta}_{kl}\left((\mb{\wt C}_{l}^\prime+\sigma^2\mb{I}_{M})^{-\frac{1}{2}}, \mc{G}_{\mb{D}_2}(z), \mb{P}_k\right)\right)_{1\leq k \leq K}.
\end{align}
where $\{\mc{G}_{\mb{D}_{1i}}(z)\}_{1\leq i \leq K}$ is the $i$-th diagonal matrix block of $\mc{G}_{\mb{D}_{1}}(z)$ of dimension $N\times N$.
\end{proposition}
\begin{proof}
	The proof of Proposition~\ref{prop3} is given in Appendix~\ref{Proof3}.
\end{proof}
Therefore, the Cauchy transform $\mc{G}_{\mb{B}}(z)$ can be obtained by (\ref{iden1}) and (\ref{equGB}). In addition, according to \cite{dobriban2020wonder}, the entries of the matrix $(z\mb{I}_{N_{\rm{tot}}}-\mb{B}_l)^{-1}$ can be approximated by the entries of $\mc{G}_{\mb{D}_1}(z)$ in~(\ref{equGB}) and the expression of the matrices $\mb{B}_{l},l=1,\dots,L$ in (\ref{matrixA}) can be then obtained.

\subsection{Cauchy transform $\mc{G}_{\mb{\wt B},\mb{\Xi}}(z)$}
Recall that $\mc{G}_{\mb{\wt{B}},\mb{\Xi}} =\frac{1}{M}\mbb{E}\left\{\mathrm{Tr}(z\mb{I}_{M}-\mb{\wt B})^{-1}\mb{\Xi}\right\}$ with $\mb{\wt{B}}=\mb{\wh H}_l\mb{P}\mb{P}^{\dag}\mb{\wh H}_l^{\dag}+\mb{\wt{C}}_{l}^{\prime}$, we adopt the similar procedure as the previous subsection. We will first resort to the linearization trick and then the operator-valued free additive convolution, which will be detailed as follows.

Based on the Anderson’s linearization trick, we can construct the block matrix $\mb{L}_{\mb{\wt B}}\in \mc{M}_{\wt n}$, where $\mc{M}_{\wt n}$ denotes the algebra of $\wt n\times \wt n$ complex random matrices and $\wt n = N_{\rm{tot}}+2 \times M$, as follows
\begin{equation} \label{C1}
	\mb{L}_{\mb{\wt B}} = 
	\left[\begin{array}{c:c}
		\mb{L}_{\mb{\wt B}}^{(1,1)} & \mb{L}_{\mb{\wt B}}^{(1,2)} \\ \hdashline
		\mb{L}_{\mb{\wt B}}^{(2,1)} & \mb{L}_{\mb{\wt B}}^{(2,2)} \\
	\end{array}\right]
		=
	\left[\begin{array}{c:ccc}
		\mb{0}_{M} & (\mb{\wt{C}}_{l}^{\prime})^{\frac{1}{2}} & \mb{\wh{H}}_l\mb{P}  \\ \hdashline
		 (\mb{\wt{C}}_{l}^{\prime})^{\frac{1}{2}}& -\mb{I}_{M} & \mb{0}_{{ M\times N_{\rm{tot}}}} \\
		\mb{P}^{\dag}\mb{\wh{H}}^{\dag}_l & \mb{0}_{ {N_{\rm{tot}}}\times M} & -\mb{I}_{{N_{\rm{tot}}}}
	\end{array}\right],
\end{equation} 
where the matrix blocks $\mb{L}_{\mb{\wt B}}^{(i,j)}$ corresponds to the partitions shown on the RHS of~(\ref{C1}).

In addition, we define the sub-algebra $\mc{D}_{\wt n}\subset \mc{M}_{\wt n}$ as the ${\wt n}\times {\wt n}$ block diagonal matrix. For each $\mb{K}_{\wt n}\in \mc{D}_{\wt n}$, it is defined as
\begin{equation} \label{K2}
	\mb{K}_{\wt n}=
	\begin{bmatrix}
		\mb{\wt D} & \mb{0}_{M} &  \mb{0}_{M \times N_{\rm{tot}}}\\ 
		\mb{0}_{M} & \mb{0}_{M} & \mb{0}_{{M\times N_{\rm{tot}}}} \\ 
		\mb{0}_{N_{\rm{tot}} \times M} & \mb{0}_{{N_{\rm{tot}}}\times M} & \mb{D}
	\end{bmatrix},
\end{equation}
where $\mb{\wt D}$ is a $M\times M$ sub-matrix and the $N_{\rm{tot}}\times N_{\rm{tot}}$ block diagonal matrix $\mb{D}$ is defined as $\mb{ D}={\mathrm{blkdiag}}\left(\mb{D}_{1},\ldots,\mb{D}_{K}\right)$ with $\{{\mb{D}_{i}}\}_{1\leq i\leq K}$ being $N\times N$ sub-matrices.

The $\mc{D}_{\wt n}$-valued Cauchy transform is defined as
\begin{equation}
	\mc{G}_{\mb{L}_{\mb{\wt{B}}}}^{\mc{D}_{\wt n}}(\mb{\Lambda}_{\wt n}(z))= \mbb{E}_{\mc{D}_{\wt n}}\left[\left(\mb{\Lambda}_{\wt n}(z) - \mb{L}_{\mb{\wt{B}}}\right)^{-1}\right],
\end{equation}
where $\mbb{E}_{\mc{D}_{\wt n}}[\mb{X}]$ has the same structure as $\mb{K}_{\wt n}$ with $\mb{X}\in \mc{M}_{\wt n}$, which is defined as
\begin{equation}
	\mbb{E}_{\mc{D}_{\wt n}} = 
	\begin{bmatrix}
		\mbb{E}[\mb{\wt D}] & \mb{0}_{M} &  \mb{0}_{M \times N_{\rm{tot}}} \\
		\mb{0}_{M} & \mb{0}_{M} & \mb{0}_{{M\times N_{\rm{tot}}}} \\
		\mb{0}_{N_{\rm{tot}} \times M} & \mb{0}_{{N_{\rm{tot}}}\times M} & \mbb{E}[\mb{D}]
	\end{bmatrix},
\end{equation} 
where $\mbb{E}[\mb{D}] = {\mathrm{blkdiag}}\left(\mbb{E}[\mb{D}_{1}],\ldots,\mbb{E}[\mb{D}_{K}]\right)$. In addition, $\mb{\Lambda}_{\wt n}(z)$ denotes the $\wt n\times \wt n$ matrix as
\begin{align}
	\mb{\Lambda}(z) = \begin{bmatrix}
		z\mb{I}_{M} & \mb{0}_{M} &  \mb{0}_{M \times N_{\rm{tot}}} \\
		\mb{0}_{M} & \mb{0}_{M} & \mb{0}_{{M\times N_{\rm{tot}}}} \\ 
		\mb{0}_{N_{\rm{tot}} \times M} & \mb{0}_{{N_{\rm{tot}}}\times M} & \mb{0}_{N_{\rm{tot}}}
	\end{bmatrix}.
\end{align}
By using the block matrix inversion lemma, we can obtain 
\begin{equation} \label{eqBL}
	\begin{aligned}
		\mc{G}_{\mb{L}_{\mb{\wt B}}}^{\mc{D}_{\wt n}}(\mb{\Lambda}_{\wt n}(z)) &=  \mbb{E}_{\mc{D}_{\wt n}}\begin{bmatrix}
			\left(z \mb{I}_{M} + \mb{L}_{\mb{\wt B}}^{(1,2)}\left(\mb{L}_{\mb{\wt B}}^{(2,2)}\right)^{-1}\mb{L}_{\mb{\wt B}}^{(2,1)}\right)^{-1} & \mb{L}_{\mb{\wt B}}^{(1,2)}\left(-z\mb{L}_{\mb{\wt B}}^{(2,2)} - \mb{L}_{\mb{\wt B}}^{(2,1)}\mb{L}_{\mb{\wt B}}^{(1,2)}\right)^{-1} \\
			\left(z\mb{L}_{\mb{\wt B}}^{(2,2)} - \mb{L}_{\mb{\wt B}}^{(2,1)}\mb{L}_{\mb{\wt B}}^{(1,2)}\right)^{-1}\mb{L}_{\mb{\wt B}}^{(2,1)} & -\left(\mb{L}_{\mb{\wt B}}^{(2,2)} + z^{-1}\mb{L}_{\mb{\wt B}}^{(2,1)}\mb{L}_{\mb{\wt B}}^{(1,2)}\right)^{-1}
		\end{bmatrix}.
	\end{aligned}
\end{equation}
In particular, the upper-left block of (\ref{eqBL}) can be explicitly written as
\begin{equation}
	\left(z \mb{I}_{M} + \mb{L}_{\mb{\wt B}}^{(1,2)}\left(\mb{L}_{\mb{\wt B}}^{(2,2)}\right)^{-1}\mb{L}_{\mb{\wt B}}^{(2,1)}\right)^{-1} = \left(z \mb{I}_{M} - \mb{\wh H}_l\mb{P}\mb{P}^{\dag}\mb{\wh H}_l^{\dag}-\mb{\wt{C}}_{l}^{\prime}\right)^{-1}.
\end{equation}
Thus, the Cauchy transform of $\mb{\wt B}$ is given by
\begin{equation}
	\mc{G}_{\mb{\wt{B}},\mb{\Xi}} =\frac{1}{M} \mathrm{tr}\left(\left\{\mc{G}_{\mb{L}_{\mb{\wt B}}}^{\mc{D}_{\wt n}}(\mb{\Lambda}_{\wt n}(z))\right\}^{(1,1)}\mb{\Xi}\right),
\end{equation}
where $\{\cdot\}^{(1,1)}$ denotes the upper-left $M\times M$ matrix block. 

Following the similar procedure, we separate $\mb{L}_{\mb{\wt B}}$ into $\mb{\ob L}_\mb{\mb{\wt B}}$ and $\mb{\wt L}_\mb{\mb{\wt B}}$, which are defined as 
\begin{align}
	\mb{\ob L} = 
	\left[\begin{array}{c:cc}
		\mb{0}_{M} & (\mb{\wt{C}}_{l}^{\prime})^{\frac{1}{2}} &  \mb{\ob{H}}_l\mb{P}\\ \hdashline
		(\mb{\wt{C}}_{l}^{\prime})^{\frac{1}{2}}& -\mb{I}_{M} & \mb{0}_{{M\times N_{\rm{tot}}}} \\ 
		\mb{P}^{\dag}\mb{\ob{H}}_l^{\dag} & \mb{0}_{{N_{\rm{tot}}}\times M} & -\mb{I}_{{N_{\rm{tot}}}}
	\end{array}\right] , \quad
	\mb{\wt L} = 
	\left[\begin{array}{c:cc}
		\mb{0}_{M} & \mb{0}_{M} &  \mb{\wt{H}}_l\mb{P}\\ \hdashline
		\mb{0}_{M} & \mb{0}_{M} & \mb{0}_{{M\times N_{\rm{tot}}}} \\ 
		\mb{P}^{\dag}\mb{\wt{H}}_l^{\dag} & \mb{0}_{{N_{\rm{tot}}}\times M} & \mb{0}_{{N_{\rm{tot}}}}
	\end{array}\right],
\end{align}
where $\mb{L}_{\mb{\wt B}}=\mb{\ob L}_\mb{\mb{\wt B}}+\mb{\wt L}_\mb{\mb{\wt B}}$.  
\begin{proposition} \label{prop4}
	The random variable $\mb{\wt L}_\mb{\mb{\wt B}}$ is semicircular and free from $\mb{\ob L}_\mb{\mb{\wt B}}$ over $\mc{D}_{\wt n}$.
\end{proposition}
\begin{proof}
	The proof of Proposition~\ref{prop4} is given in Appendix~\ref{Proof4}.
\end{proof}

Based on Proposition~\ref{prop4}, the limiting spectral distribution of $\mb{L}_\mb{\wt B}$ can be determined by the operator-valued free additive convolution of $\mb{\ob{L}}_{\mb{\wt B}}$ and $\mb{\wt L}_\mb{\wt B}$. Therefore, the operator-valued Cauchy transform $\mc{G}_{\mb{L}_{\mb{\wt B}}}^{\mc{D}_{\wt n}}(\mb{\Lambda}_{\wt n}(z))$ can be obtained by using the subordination formula, which is given as follows
\begin{equation}\label{sub2}
	\mc{G}_{\mb{L}_{\mb{\wt B}}}^{\mc{D}_{\wt n}}(\mb{\Lambda}_{\wt n}(z))= \mbb{E}_{\mc{D}_{\wt n}}\left[\left(\mb{\Lambda}_{\wt n}(z) - \mc{R}_{\widetilde{\mb{L}}_{\wt B}}^{\mc{D}_{\wt n}}\left(\mc{G}_{\mb{L}_{\mb{\wt B}}}^{\mc{D}_{\wt n}}(\mb{\Lambda}_{\wt n}(z))\right) - \ob{\mb{L}}_{\mb{\wt B}}\right)^{-1}\right],
\end{equation}
where $\mc{R}_{\wt{\mb{L}}_{\mb{\wt B}}}^{\mc{D}_{\wt n}}(\cdot)$ denotes the $\mc{D}_{\wt n}$-valued $R$-transform of ${\mb{L}}_{\mb{\wt B}}$. Then, the expression of  the Cauchy transform $\mc{G}_{\mb{\wt{B}},\mb{\Xi}}$ can be obtain by the following proposition.
\begin{proposition}\label{prop5}
	The expression of Cauchy transform $\mc{G}_{\mb{\wt B},\mb{\Xi}}(z)$ is given by
	\begin{equation}
		\mc{G}_{\mb{\wt{B}},\mb{\Xi}} =\frac{1}{M} \mathrm{tr}\left(\mc{G}_{\mb{\wt D}}(z)\mb{\Xi}\right),
	\end{equation}
where
\begin{align}
	\mc{G}_{\mb{\wt D}}(z) &= \left( \mb{\Phi}(z)- \mb{\wt{C}}_{l}^{\prime} -  \mb{\ob{H}}_l\mb{P}\mb{\wt \Phi}^{-1}(z)\mb{P}^{\dag}\mb{\ob{H}}_l^{\dag}\right)^{-1}, \label{equGD}\\
	\mc{G}_{\mb{D}}(z)&=\left( \mb{\wt \Phi}(z)-\mb{P}^{\dag}\mb{\ob{H}}_l^{\dag}\mb{\Phi}^{-1}(z)\mb{\ob{H}}_l\mb{P}-\mb{\Upsilon}(z)\right)^{-1}, \label{equGDT}
\end{align}
The matrix-valued function $\mb{\Phi}(z)$, $\mb{\wt \Phi}(z)$ and $\mb{\Upsilon}(z)$ are respectively denoted as 
\begin{align}
	\mb{\Phi}(z)&=z\mb{I}_{M}-\sum_{i=1}^{K}\bm{\zeta}_{i}({{\mc{G}_{\mb{D}_{i}}(z)}}),\\
	\mb{\wt \Phi}(z)&=\mb{I}_{N_{\rm{tot}}}-\bm{\wt \zeta}(\mb{\mc{G}_{\mb{\wt D}}(z)}),\\
	\mb{\Upsilon}(z) &= \mb{P}^{\dag}\mb{\ob{H}}_l^{\dag}\mb{\Phi}^{-1}(z)(\mb{\wt{C}}_{l}^{\prime})^{\frac{1}{2}}\left( \mb{I}_{M}-(\mb{\wt{C}}_{l}^{\prime})^{\frac{1}{2}}\mb{\Phi}^{-1}(z)(\mb{\wt{C}}_{l}^{\prime})^{\frac{1}{2}}\right)^{-1}(\mb{\wt{C}}_{l}^{\prime})^{\frac{1}{2}}\mb{\Phi}^{-1}(z)\mb{\ob{H}}_l\mb{P},
\end{align}
where the functions of one-sided correlation matrices are given by
\begin{align}
	\bm{\wt \zeta}(\mc{G}_{\mb{\wt D}}(z))&={\rm{blkdiag}}\left(\wt{\eta}_{1l}(\mb{I}_{M}, \mc{G}_{\mb{\wt D}}(z), \mb{P}_1),\ldots, \wt{\eta}_{Kl}(\mb{I}_{M},\mc{G}_{\mb{\wt D}}(z),\mb{P}_K)\right),\\
	\bm{\zeta}_{i}(\mc{G}_{\mb{D}_i}(z))&={\eta}_{il}(\mb{I}_{M},\mc{G}_{\mb{D}_i}(z),\mb{P}_i),
\end{align}
and $\{\mc{G}_{\mb{D}_i}(z)\}_{1\leq i \leq K}$ is the $i$-th diagonal matrix block of $\mc{G}_{\mb{D}}(z)$ of dimension $N\times N$. 
\end{proposition}
\begin{proof}
	The proof of Proposition~\ref{prop5} is given in Appendix~\ref{Proof5}.
\end{proof}

Based on Proposition~\ref{prop4} and Proposition~\ref{prop5}, the derivatives of Cauchy transform $\mc{G}^\prime_{\mb{B}_l}(z)$ and $\mc{G}^\prime_{\mb{\wt{B}}_l, \mb{S}_l\mb{S}_l^{\dag}}(z)$ can be obtained by matrix calculus, which is given in Appendix~\ref{Appderivative}. In this way, the explicit expressions of $\mb{v}$, $\mb{A}$ and $\mb{Y}$ can be obtained and the CPU-UAV can calculate the combining weights with the channel statistics, which mitigates the interaction of information.

\section{Numerical Results} \label{secResult}
In the above sections, we proposed the MMSE-OneShot scheme and utilized the operator-valued Cauchy transforms to obtain the asymptotic expressions of the combining weights. In this section, we numerically evaluate the performance of the proposed MMSE-OneShot scheme, where both the explicit and the asymptotic combining weights are calculated and compared. In addition, we investigate the SE of the aerial cell-free systems with various numbers of deployed AP-UAVs and numbers of equipped antennas. The fully-centralized scheme in~\ref{FCScheme} and the small-cell network are simulated as the benchmarks. In small-cell networks, each UE selects the nearest AP-UAV to transmit its signals, and the detection is only implemented in the selected AP-UAV~\cite{wang2022uplink}.

\subsection{Simulation Setup}
We assume that the AP-UAVs are regularly deployed in a square area of $1\times 1$~$\text{km}^2$, and UEs are randomly distributed within this area. The transmitting power of each UE is set as $p_k=23\text{dBm}$ and the noise power in each AP-UAV is $\sigma^2=-94\text{dBm}$. The precoder of each UE is assumed to be a unit matrix, which means that the transmitting power is divided equally between each antenna, i.e. $\mb{F}_k=\mb{P}_k=\sqrt{\frac{p_k}{N}}\mb{I}_N$. The AP-UAVs are assumed to hover at the altitude of $h=100$m, unless otherwise stated, and the height of UE is set as $1.5$m. We consider that each AP-UAV has the same number of antennas, and each coherence block contains $\tau_c = 200$ channel uses and $\tau_p = KN$, unless otherwise stated. 

We consider that the air-to-ground channels consist of LoS and NLoS links with different probability~\cite{alzenad20173}. In specific, the probability of having an LoS connection between the UE $k$ and AP-UAV $l$ is given by
\begin{equation}
	P^{\text{LoS}}_{kl}=\frac{1}{1+a\exp(-b(\frac{180}{\pi}\tan^{-1}(\frac{h}{d_{kl}})-a))},
\end{equation}
where $a$ and $b$ are constants depending on the environments (suburban, urban, and dense urban)~\cite{al2014optimal}, and $d_{kl}$ denotes the distance between AP-UAV $l$ and UE $k$. The probability of NLoS is then denoted by $P^{\text{NLoS}}_{kl}=1-P^{\text{LoS}}_{kl}$.

The path loss model for LoS and NLoS links in dB are respectively
\begin{align}
	\text{PL}_{kl}^{\text{LoS}} &= \text{FSPL}_{kl}+ \varrho_{\text{LoS}},\\
	\text{PL}_{kl}^{\text{NLoS}} &= \text{FSPL}_{kl}+ \varrho_{\text{NLoS}},
\end{align}
where $\text{FSPL}_{kl}=10\log(\beta_{kl})$ is the free space path loss between AP-UAV $l$ and UE $k$ with $\beta_{kl}=\rho/d_{kl}^\alpha$, where $\rho=-55\text{dB}$ is the path loss constant, and $\alpha=3$ is the path loss exponent. The excessive path loss $\varrho_{\text{LoS}}$ and $\varrho_{\text{NLoS}}$ depend on the propagation group (LoS and NLoS) and are given in~\cite{al2014optimal}. Therefore, the averaged path loss between AP-UAV $l$ and UE $k$ is
\begin{equation}
	\text{PL}_{kl} = P^{\text{LoS}}_{kl}\times \text{PL}_{kl}^{\text{LoS}} + P^{\text{NLoS}}_{kl} \times \text{PL}_{kl}^{\text{NLoS}}.
\end{equation}


In addition, according to~\cite{zhang2013capacity}, the spatial correlation matrices $\mb{R}_{kl}$ and $\mb{T}_{kl}$ are generated from a uniform linear array with half wavelength spacing, whose $(m,n)$-th element are given as follows
\begin{equation}
	\left[\text{T}~(\text{or}~\mathbf{R})_{kl}\right]_{mn}=\int_{-180}^{180}\frac{d\phi}{\sqrt{2\pi\xi_{kl}^2}}e^{j\pi(m-n)\sin\left(\frac{\pi\phi}{180}\right)-\frac{(\phi-\theta_{kl})^2}{2\xi_{kl}^2}},
\end{equation}
where $\theta_{kl}$ and $\xi_{kl}$ are the mean angle and angular standard deviation (ASD), respectively. The LoS component is generated according to $\mb{\ob{H}}_{kl}=\mb{a}_{{\rm{R}},l}(\theta^{{{R}}}_{kl})\mb{a}_{{\rm{T}},k}(\theta^{{{T}}}_{kl})^{\dag}$, where 
\begin{align}
	\mb{a}_{{{R}},l}(\theta^{{{R}}}_{kl}) &
	=\begin{bmatrix}
		1 & e^{j\pi\sin\left(\frac{{\theta}_{kl}^{{R}}}{180}\pi\right)}&
		\ldots&
		e^{j\pi(M-1)\sin\left(\frac{{\theta}_{kl}^{{R}}}{180}\pi\right)}
	\end{bmatrix}^{\rm T}, \\
	\mb{a}_{{{T}},k}(\theta^{{{T}}}_{kl}) &
	=\begin{bmatrix}
		1 & e^{j\pi\sin\left(\frac{{\theta}_{kl}^{{T}}}{180}\pi\right)}&
		\ldots&
		e^{j\pi(N-1)\sin\left(\frac{{\theta}_{kl}^{{T}}}{180}\pi\right)}
	\end{bmatrix}^{\rm T},
\end{align}
where $\theta^{{{R}}}_{kl}$ and $\theta^{{{T}}}_{kl}$ denote the angles of LoS path at AP-UAV $l$ and UE $k$, respectively.

\subsection{Accuracy Analysis}
We first investigate the accuracy of our derived asymptotic expression of the combining weights with $L=8$ and $K=4$, as illustrated in Fig.~\ref{accuracy}. The Monte Carlo simulation curves are obtained by averaging the sum SE over $10^5$ channel realizations and plotted as markers. As observed in Fig.~\ref{accuracy}, the sum SE monotonously increases with the number of antennas per UE and AP-UAV. The theoretical result fits the Monte Carlo simulation curve accurately, even when UEs and AP-UAVs are equipped with a small number of antennas. 
\begin{figure}[htb]
	\centerline{\includegraphics[width=0.5\columnwidth]{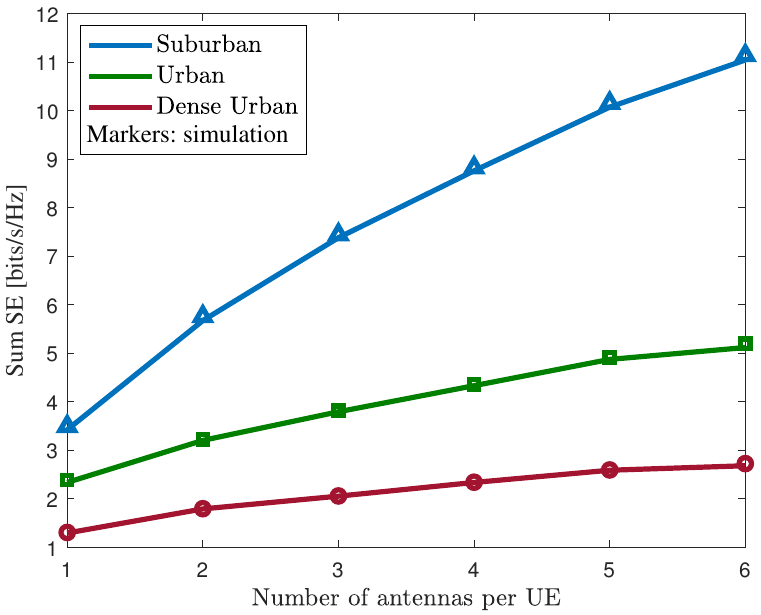}}
	\caption{Sum SE versus the number of antennas per UE and per AP-UAV with a fixed ratio, i.e., $M/N=2$, in different urban environments. The solid lines represent theoretical results and the markers represent Monte Carlo simulation results. In each case, the number of AP-UAV and UE are $L=8$ and $K=4$, respectively}
	\label{accuracy}
\end{figure}

\subsection{Effects of Distributed Deployment}
In practice, the distributed deployment of AP-UAVs, including the location, the number of AP-UAVs, and the number of antennas per AP-UAV, etc, has a significant influence on the system performance. In order to investigate how to determine the number of deployed AP-UAVs and the number of antennas per AP-UAV, we simulate the performance of the proposed MMSE-OneShot scheme in distributed systems under a given total number of receiving antennas. In Fig.~\ref{Ratio_of_LM}, we plot the sum SE as a function of the ratio of the number of AP-UAVs to the number of antennas per AP-UAV for a fixed total number of receiving antennas $M_{\rm{tot}}=144$. Note that only one AP-UAV is deployed when the ratio $L/M=1/144$ and the number of AP-UAVs gets larger as the ratio gets larger, while the number of antennas per AP-UAV gets smaller. We notice that the sum SE can reach a maximum value when $L/M=1$ in the suburban environment. Whereas in urban and dense urban environments, the sum SE continuously increases with the number of AP-UAVs getting larger as long as the number of antennas per AP-UAV is greater than one. This is due to the fact that there is a tradeoff between the LoS probability and spatial diversity. In specific, the LoS probability can be increased by deploying more distributed AP-UAVs, which also achieves macro-diversity. Nevertheless, deploying more antennas in every AP-UAV can take full advantage of the spatial diversity. In addition, the LoS probability in the suburban environment is greater than that in urban and dense urban environments~\cite{al2014optimal}. Therefore, it is more effective to increase the LoS probability (i.e., increase the number of AP-UAVs) in urban and dense urban environments, while system performance can be improved by increasing the spatial diversity (i.e., increasing the number of antennas per AP-UAV) in the suburban environment.

\begin{figure}[htb]
	\centerline{\includegraphics[width=0.5\columnwidth]{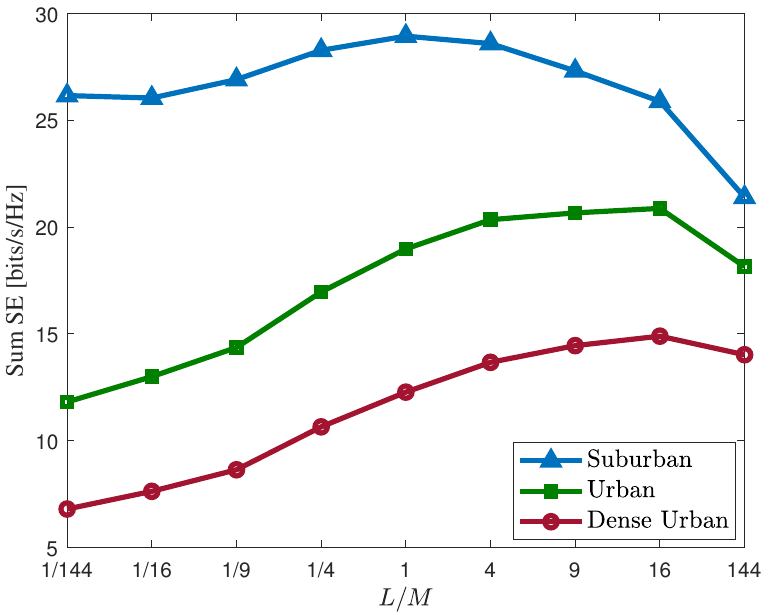}}
	\caption{Sum SE versus varying ratio of the number of AP-UAVs to the number of antennas per AP-UAV, i.e., $L/M$, in different urban environments with a fixed total number of transmitting antennas $M_{\rm{tot}}=144$. In each case, the number of UEs and antennas per AP-UAV are $K=16$ and $N=2$, respectively.}
	\label{Ratio_of_LM}
\end{figure}

\begin{figure}[htb]
	\centerline{\includegraphics[width=0.5\columnwidth]{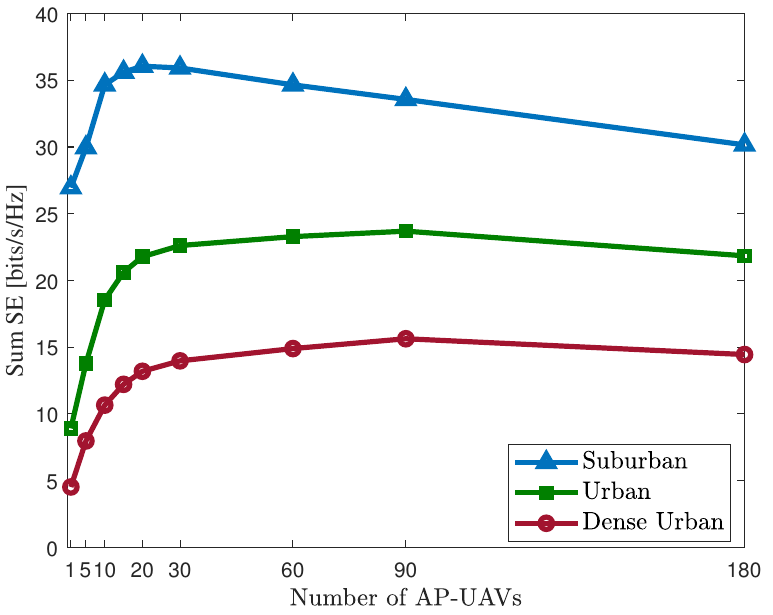}}
	\caption{Sum SE versus the number of AP-UAVs with a fixed total number of transmitting antennas $M_{\rm{tot}}=180$ in different urban environments. In each case, the number of UEs and antennas per AP-UAV are $K=10$ and $N=4$, respectively.}
	\label{NUAV}
\end{figure}

 The deployment strategy is further illustrated in Fig.~\ref{NUAV}, where the total number of receiving antennas is fixed as $M_{\rm{tot}}=180$. In the suburban environment, the sum SE significantly decreases when the number of antennas per AP-UAV is smaller than $9$ (when $L>20$). However, in urban and dense urban environments, the sum SE continuously increases until $L=180$, i.e., every AP-UAV is equipped with a single antenna. This phenomenon also indicates the tradeoff between the LoS probability and spatial diversity and that we should allocate the number of AP-UAVs and the number of receiving antennas per AP-UAV according to the physical propagation environment. 
 
 \subsection{Effects of the Height of AP-UAV in different environments}
 As shown in Fig.~\ref{Ratio_of_LM} and Fig.~\ref{NUAV}, the performance achieved by deploying more AP-UAVs is always better than that by deploying only one AP-UAV with multiple antennas, and it's necessary to establish a distributed system by deploying multiple AP-UAVs. Thus, we further discuss the different types of distributed systems in the following. 
 The effect of the height of AP-UAV in different urban environments is investigated in Fig.~\ref{Height} with $L=8$, $K=4$, $M=8$, and $N=4$. The fully-centralized scheme achieves the highest sum SE in all considered environments, while the small-cell scheme gets the worst performance. The proposed MMSE-OneShot scheme can obtain similar performance as the fully-centralized scheme, especially in urban and dense urban environments. In addition, we notice that as the height of AP-UAVs gets higher, the sum SE initially increases to reach a maximum value and then starts to decrease in urban and dense urban environments. The reason for this phenomenon is that, as the height is relatively small and becomes larger, the gain from increasing LoS probability outweighs the deterioration of the path loss caused by increasing distance. When the height exceeds $100$ meters, the gain obtained from the LoS probability becomes saturated and the large path loss degrades the system performance.
 
 \begin{figure}[htb]
 	\centerline{\includegraphics[width=0.5\columnwidth]{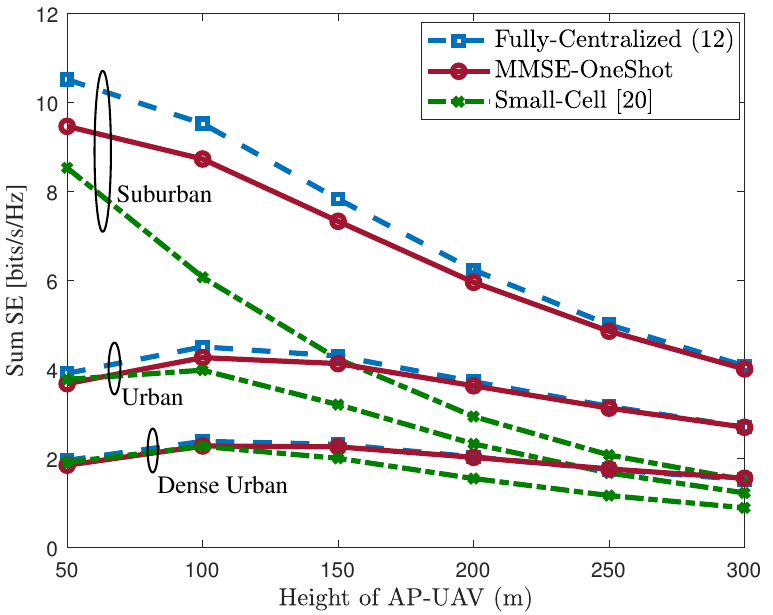}}
 	\caption{Sum SE versus the height of AP-UAVs in different urban environments. In each case, the number of AP-UAVs, UEs, antennas per AP-UAV, and antennas per UE are $L=8$, $K=4$, $M=8$, and $N=4$.}
 	\label{Height}
 \end{figure}

\begin{figure}[h]
	\centerline{\includegraphics[width=0.5\columnwidth]{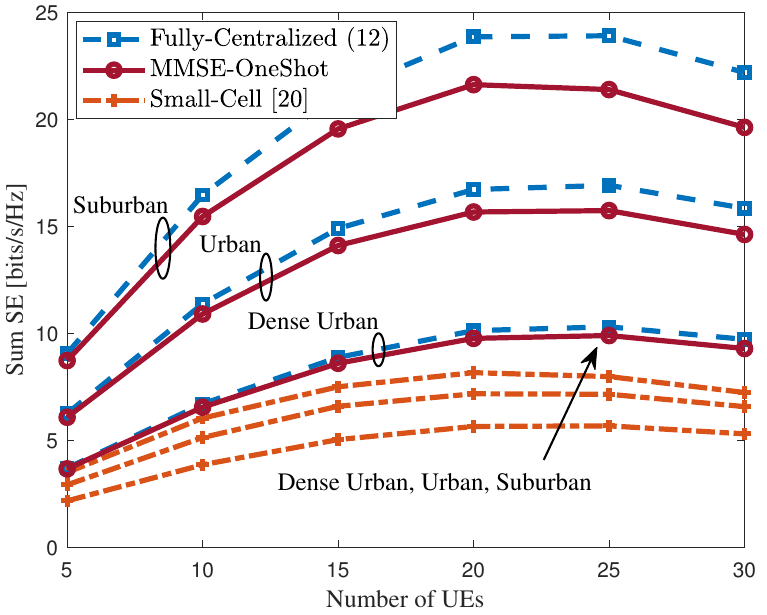}}
	\caption{Sum SE verse the number of UEs in different urban environments. In each case, the number of AP-UAVs, antennas per AP-UAV, and antennas per UE are $L=8$, $M=16$, and $N=4$.}
	\label{NUE}
\end{figure}

\begin{figure}[b]
	\centerline{\includegraphics[width=0.5\columnwidth]{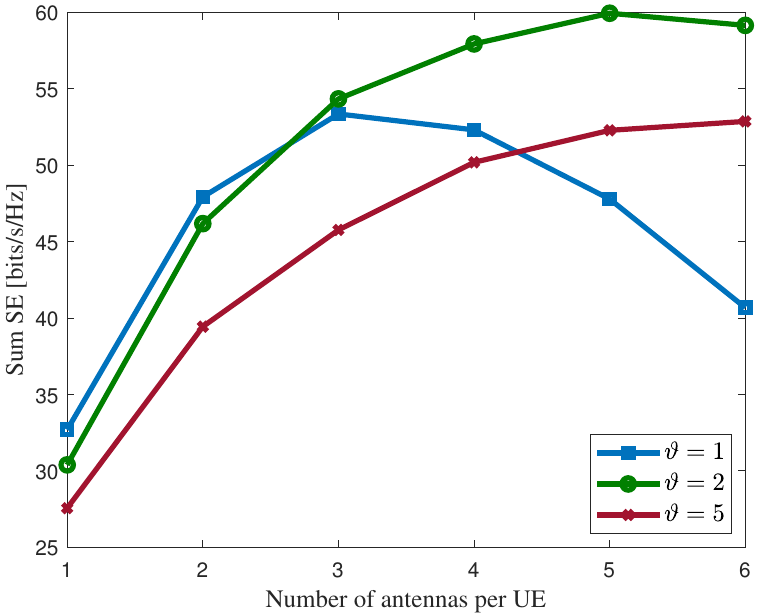}}
	\caption{Sum SE verse the number of transmit antennas with different pilot-reuse factors in the suburban environment with $L=12$, $K=10$, and $M=16$. The length of channel uses for channel estimation $\tau_p=KN/\vartheta$ and $\tau_c=100$. }
	\label{Npilot}
\end{figure}
 
\subsection{Effects of the Number of UEs}
In Fig.~\ref{NUE}, we show the sum SE as a function of the number of UEs $K$ with $L=8$, $M=16$, and $N=4$. Note that the sum SE can reach a maximum value in a specific $K$ in any case. This phenomenon is due to the pre-log factor $1-\tau_p / \tau_c$ and the inter-UE interference. With the increase of $K$, the pre-log factor in SE expressions decreases, which means that most channel uses of coherence block are used for channel estimation and the ratio of signal transmission is reduced. In addition, the inter-UE interference becomes large as the number of UEs increases, which is eventually harmful to the sum SE. Therefore, as shown in Fig.~\ref{NUE}, the optimal sum SE can be obtained by setting the active number of UEs as $20$ to $25$ in the considered system settings.

\subsection{Effects of the Length of Pilots}
To investigate the effects of the length of pilots of the proposed MMSE-OneShot scheme, we set different pilot-reuse factor $\vartheta$ in Fig.~\ref{Npilot}, where the number of pilot $\tau_p = KN/\vartheta$. The pilot reuse factor $\vartheta$ means that every $\vartheta$ UEs are allocated to the same pilot matrix and $\vartheta>1$ leads to the pilot contamination. The other parameters are set as $K=10$, $L=12$, and $M=16$. We observe that there is a maximum sum SE when $\vartheta=1$ and $\vartheta=2$ at $N=3$ and $N=5$, respectively. When the pilot-reuse factor is small (higher $\tau_p$), increasing the number of transmitting antennas requires a large number of pilots, which decreases the signal transmission efficiency and sum SE. On the other hand, when the pilot-reuse factor is large (lower $\tau_p$), a smaller number of pilots are needed and the pilot contamination gets worse. However, the gain from the increasing of transmitting antennas exceeds the loss of pilot contamination and thus the system performance monotonically increases as the transmitting antennas get large. Therefore, there is a trade-off between the pilot-reuse factor and the number of transmitting antennas and it is necessary to set the pilot-reuse factor according to the number of antennas per UE.

\section{Conclusion} \label{secConclude}
Limited by the lightweight fuselage, several AP-UAVs with few antennas compose the CF mMIMO system, where all AP-UAVs forward their channel coefficients to the CPU-UAV. However, due to the limited capacity of the wireless fronthaul link, the huge interaction overhead prevents the aerial CF mMIMO system from being deployed in practice. 
To address the above challenges, we propose a two-layer distributed uplink signal detection scheme for the CF mMIMO system. Specifically, the uplink signals are first recovered in AP-UAVs with local MMSE detector and then forwarded to the CPU-UAV, which weighted combines the signals to obtain the refined signals. By using the operator-valued free probability theory, the asymptotic expressions of the combining weights are obtained explicitly for the large dimensional antenna configurations, while the results are also verified to be accurate in realistic system configurations. Numerical results show that  the proposed MMSE-OneShot scheme not only achieves similar system performance as the fully-centralized scheme in various environments, but also significantly alleviates the interactive overhead between the AP-UAVs and the CPU-UAV. In addition, we investigate the impact of the number of deployed AP-UAVs as well as the number of receive antennas per AP-UAV on the achieved sum SE. In the suburban environment, where the LoS probability is relatively large, the number of antennas per AP-UAV should be increased to achieve spatial diversity. While in urban and dense urban environments, more AP-UAVs should be deployed to increase the LoS probability. 



\begin{appendices}
\section{Proof of Proposition~\ref{prop2}} \label{Proof2}
According to \cite{mingo2017free}, a random variable $\widetilde{\mb{L}}\in\mc{M}$ is said to be $\mc{D}$-valued semicircular if the free cumulant 
\begin{align}
	\kappa_m^{\mc{D}}(\widetilde{\mb{L}} b_1, \widetilde{\mb{L}} b_2, \ldots, \widetilde{\mb{L}} b_{m-1}, \widetilde{\mb{L}}) = 0,
\end{align}
for all $m\neq 2$, and all $b_1,\ldots,b_{m-1}\in\mc{D}$. The free cumulant $\kappa_m^{\mc{D}}$ is a mapping from $\mc{M}_m$ to $\mc{D}$.

We first expend $\mb{\wt L}_{\mb{B}}$ into a summation of $n\times n$ sub-matrices, such that
\begin{equation}
	\mb{\wt L}_{\mb{B}} =\sum_{k=1}^{K}\mb{\wt L}_{\mb{B},k},
\end{equation}
where the sub-matrices $\mb{\wt L}_{\mb{B},k}$ are given by
\begin{align} \label{L_1}
	\mb{\wt L}_{\mb{B},k} = 
	\begin{bmatrix}
		\mb{0}_{N_{\rm{tot}}} & \mb{\Sigma}_k^{\dag}\\
		\mb{\Sigma}_k & \mb{0}_{M}
	\end{bmatrix} ,
\end{align}
and $\mb{\Sigma}_k = \left[\mb{0}_{M\times N}, \ldots, (\mb{\wt C}_{l}^\prime+\sigma^2\mb{I}_{M})^{-\frac{1}{2}}\mb{\wt H}_{kl}\mb{P}_k,\ldots,\mb{0}_{M\times N}\right]\in\mbb{C}^{M\times N_{\rm{tot}}}$.

Recall that $\mb{\wh{C}}_{kl}$ is the correlation matrix of $\mb{\wt h}_{kl}$ and $\mb{\wt h}_{kl}=\text{vec}(\mb{\wt H}_{kl})$, thus we have
\begin{equation}
	\mb{\wt H}_{kl} = (\mb{\check{C}}_{kl})^{\frac{1}{2}}\mb{X}_{kl},
\end{equation}
where $\mb{\check{C}}_{kl} =\frac{1}{N} \sum_{i=1}^{N}\mb{\wh{C}}_{kl}^{ii}\in\mbb{C}^{M\times M}$ with $\mb{\wh{C}}_{kl}^{ii}$ being the $i$-th diagonal sub-matrix of $\mb{\wh{C}}_{kl}$, and the random matrix $\mb{X}_{kl}$ are $i.i.d.$ complex Gaussian distributed with entries having zero mean and unit variance. 
Therefore, the sub-matrices $\mb{\wt L}_{\mb{B},k}$ can be rewritten as 
\begin{equation}
	\mb{\wt L}_{\mb{B},k} = \mc{A}_k\mc{X}_k\mc{A}_k^{\dag},
\end{equation}
where $\mc{X}_k$ has the same structure as the block matrix $\mb{\wt L}_{\mb{B},k}$ in~(\ref{L_1}), while replacing the matrix $(\mb{\wt C}_{l}^\prime+\sigma^2\mb{I}_{M})^{-\frac{1}{2}}\mb{\wt H}_{kl}\mb{P}_k$ with $\mb{X}_{kl}$. The block matrix $\mc{A}_k$ is given by
\begin{align}
	\mc{A}_k &= 
	\begin{bmatrix}
		\mb{\wh P}_k^{\dag} & \mb{0}_{N_{\rm{tot}} \times M} \\
		\mb{0}_{M \times N_{\rm{tot}}} & (\mb{\wt C}_{l}^\prime+\sigma^2\mb{I}_{M})^{-\frac{1}{2}}\mb{\check{C}}_{kl}^{\frac{1}{2}} \\
	\end{bmatrix},
\end{align}
where $\mb{\wh P}_k = \text{blkdiag}(\mb{0}_{N},\dots,\mb{P}_k,\dots, \mb{0}_{N})$ with $\mb{P}_k$ in its $k$-th diagonal block.

Since $\{\mc{X}_k\}_{1\leq k\leq K}$ are Wigner matrices and independent from each other, they are semicircular and free over the sub-algebra $\mc{D}_n$. Following the same arguments as in \cite[Appendix B]{lu2016free}, $\{\mb{\wt L}_{\mb{B},k}\}_{1\leq k\leq K}$ are semicircular and free over the subalgebra $\mc{D}_n$. Therefore, the summation of $\{\mb{\wt L}_{\mb{B},k}\}$ is also semicircular over the subalgebra $\mc{D}_n$ and is free from any deterministic matrix from $\mc{M}_n$.

\section{Proof of Proposition~\ref{prop3}} \label{Proof3}
As previously stated, the limiting spectral distribution of $\mb{L}_{\mb{B}}$ is the free convolution of $\mb{\ob{L}}_{\mb{B}}$ and $\mb{\wt L}_{\mb{B}}$. The $\mc{D}_n$-valued Cauchy transform $\mc{G}_{\mb{L}_{\mb{B}}}^{\mc{D}_n}(\mb{\Lambda}_n(z))$ can be obtained by using the subordination formula~(\ref{sub1}). Recall that the $R$-transform $\mc{R}_{\widetilde{\mb{L}}_{\mb{B}}}^{\mc{D}_n}(\cdot)$ is the free cumulant generating function of $\mb{\wt L}_{\mb{B}}$ with the following formal power series expansion:
\begin{equation}
	\mc{R}_{\widetilde{\mb{L}}_{\mb{B}}}^{\mc{D}_n}\left(\mb{K}_n\right) = \kappa_1^{\mc{D}_n}({\mb{K}_n}) + \kappa_2^{\mc{D}_n}(\widetilde{\mb{L}}_{\mb{B}}\mb{K}_n,\widetilde{\mb{L}}_{\mb{B}}) + \kappa_3^{\mc{D}_n}(\widetilde{\mb{L}}_{\mb{B}}\mb{K}_n,\widetilde{\mb{L}}_{\mb{B}}\mb{K}_n,\widetilde{\mb{L}}_{\mb{B}}) + \cdots,\label{eqRBL}
\end{equation}
where $\kappa_i^{\mc{D}_n}$ denotes the $i$-th free cumulant of $\widetilde{\mb{L}}_{\mb{B}}$ over $\mc{D}_n$. In addition, since $\widetilde{\mb{L}}_n$ is semicircular over $\mc{D}_n$, all its cumulants in (\ref{eqRBL}) except $\kappa_2^{\mc{D}_n}$ are zero. 
Therefore, the $R$-transform reduces to the covariance function of $\wt{\mb{L}}_{\mb{B}}$ over $\mc{D}_n$ parameterized by $\mb {K}_n$, i.e.,
\begin{align} \label{R1}
	\mc{R}_{\widetilde{\mb{L}}_{\mb{B}}}^{\mc{D}_n}(\mb{K}_n) &= \mbb{E}_{\mc{D}_n}\left[\widetilde{\mb{L}}_{\mb{B}}\mb{K}_n\widetilde{\mb{L}}_{\mb{B}}\right],\nonumber\\
	&=
	\begin{bmatrix}
		\wt{\delta}(\mb{D}_2) & \mb{0} \\
		\mb{0} & \sum_{i=1}^{K}{\delta}_i(\mb{D}_{1i}) 
	\end{bmatrix},
\end{align}
where $\wt{\delta}(\mb{D}_2)={\rm{blkdiag}}\left(\wt{\eta}_{1l}((\mb{\wt C}_{l}^\prime+\sigma^2\mb{I}_{M})^{-\frac{1}{2}}, \mb{D}_2, \mb{P}_1),\ldots,\wt{\eta}_{Kl}((\mb{\wt C}_{l}^\prime+\sigma^2\mb{I}_{M})^{-\frac{1}{2}},\mb{D}_2,\mb{P}_K)\right)$ and ${\delta}_i(\mb{D}_{1i})={\eta}_{il}((\mb{\wt C}_{l}^\prime+\sigma^2\mb{I}_{M})^{-\frac{1}{2}},\mb{D}_{1i},\mb{P}_i)$.

We divide $\mc{G}_{\mb{L}_{\mb{B}}}^{\mc{D}_n}(\mb{\Lambda}_n(z))\in \mc{D}_n$ into the same partition as (\ref{K1}), $\mc{G}_{\mb{L}_{\mb{B}}}^{\mc{D}_n}(\mb{\Lambda}_n(z))$ can be rewritten as 
\begin{align} \label{partition1}
	\mc{G}_{\mb{L}_{\mb{B}}}^{\mc{D}_n}(\mb{\Lambda}_n(z)) = 
	\begin{bmatrix}
		\mc{G}_{\mb{D}_1}(z) & \mb{0}\\
		\mb{0} & \mc{G}_{\mb{D}_2}(z)
	\end{bmatrix},
\end{align} 
where $\mc{G}_{\mb{D}_1}(z)={\rm{blkdiag}}(\mc{G}_{\mb{D}_{11}}(z), \ldots, \mc{G}_{\mb{D}_{1K}}(z))\in\mbb{C}^{N_{\rm{tot}}\times N_{\rm{tot}}}$ with $\mc{G}_{\mb{D}_{1k}}(z)$ being a $N\times N$ sub-matrix and $\mc{G}_{\mb{D}_2}(z)$ is a $M\times M$ matrix. By replacing $\mb{K}_n$ in~(\ref{R1}) with $\mc{G}_{\mb{L}_{\mb{B}}}^{\mc{D}_n}(\mb{\Lambda}_n(z))$ and substituting $\mc{R}_{\widetilde{\mb{L}}_{\mb{B}}}^{\mc{D}_n}\left(\mc{G}_{\mb{L}_{\mb{B}}}^{\mc{D}_n}(\mb{\Lambda}_n(z))\right)$ and $\mb{\ob L}_{\mb{B}}$ in the subordination formula~(\ref{sub1}), we have
\begin{align}
	\mc{G}_{\mb{L}_{\mb{B}}}^{\mc{D}_n}(\mb{\Lambda}_n(z)) = \mbb{E}_{\mc{D}_n}
	\begin{bmatrix}
		\left(\mb{\Psi}(z)-\mb{\ob G}^{\dag}\wt{\mb{\Psi}}^{-1}(z)\mb{\ob G}\right)^{-1} & \mb{\Psi}^{-1}(z)\mb{\ob G}^{\dag}\left(\wt{\mb{\Psi}}-\mb{\ob G}\mb{\Psi}^{-1}(z)\mb{\ob G}^{\dag}\right)^{-1}\\
		\left(\wt{\mb{\Psi}}-\mb{\ob G}\mb{\Psi}^{-1}(z)\mb{\ob G}^{\dag}\right)^{-1}\mb{\ob G}\mb{\Psi}^{-1}(z) & \left(\mb{\wt \Psi}(z)-\mb{\ob G}{\mb{\Psi}}^{-1}(z)\mb{\ob G}^{\dag}\right)^{-1}\\
	\end{bmatrix},
\end{align}
where $\mb{\Psi}(z)=z\mb{I}_{N_{\rm{tot}}}-\wt{\delta}(\mc{G}_{\mb{D}_2}(z))$ and $\mb{\wt{\Psi}}=\mb{I}_{M}-\sum_{i=1}^{K}\delta_i(\mc{G}_{\mb{D}_{1i}}(z))$. Thus, based on the partition~(\ref{partition1}), the expressions of $\mc{G}_{\mb{D}_1}(z)$ and $\mc{G}_{\mb{D}_2}(z)$ can be then obtained.

\section{Proof of Proposition~\ref{prop4}} \label{Proof4}
Following the similar procedure as in Appendix~\ref{Proof2}, we expend $\mb{\wt L}_{\mb{\wt B}}$ into a summation of $\wt n\times \wt n$ sub-matrices, such that
\begin{equation}
	\mb{\wt L}_{\mb{\wt B}} = \sum_{k=1}^{K}\mb{\wt L}_{\mb{\wt B},k},
\end{equation}
where the sub-matrices $\mb{\wt L}_{\mb{\wt B},k}$ are given by
\begin{equation} \label{L2}
	\mb{\wt L}_{\mb{\wt B},k} = 
	\left[\begin{array}{c:cc}
		\mb{0}_{M} & \mb{0}_{M} &  \mb{\wt \Sigma}_k\\ \hdashline
		\mb{0}_{M} & \mb{0}_{M} & \mb{0}_{{M\times N_{\rm{tot}}}} \\
		\mb{\wt \Sigma}_k^{\dag} & \mb{0}_{{N_{\rm{tot}}}\times M} & \mb{0}_{{N_{\rm{tot}}}}
	\end{array}\right],
\end{equation}
and the embedded matrix $\mb{\wt \Sigma}_k=\left[\mb{0}_{M\times N},\dots,\mb{\wt H}_{kl}\mb{P}_k,\dots, \mb{0}_{M\times N}\right]\in \mbb{C}^{M\times N_{\rm{tot}}}$.

We rewrite $\mb{\wt L}_{\mb{\wt B},k}$ as follows
\begin{equation}
	\mb{\wt L}_{\mb{\wt B},k} = \mc{\wt X}_k \mc{\wt A}_k \mc{\wt X}_k^{\dag},
\end{equation} 
where $\mc{\wt X}_k$ has the same structure as $\mb{\wt L}_{\mb{\wt B},k}$ in~(\ref{L2}), while replacing the matrix $\mb{\wt H}_{kl}\mb{P}_k$ with $\mb{X}_{kl}$. In addition, the block matrix $\mc{\wt A}_k$ is expressed as
\begin{equation}
	\mc{\wt A}_k = 
	\begin{bmatrix}
		\mb{\check{C}}_{kl} & \mb{0}_{M} & \mb{0}_{M\times N_{\rm{tot}}} \\
		\mb{0}_{M} & \mb{0}_{M} & \mb{0}_{M\times N_{\rm{tot}}} \\
		\mb{0}_{N_{\rm{tot}}\times M} & \mb{0}_{N_{\rm{tot}}\times M} & \mb{\wh P}_k \\
	\end{bmatrix}
\end{equation}
where $\mb{\wh P}_k = \text{blkdiag}(\mb{0}_{N},\dots,\mb{P}_k,\dots, \mb{0}_{N})$. 
Therefore, following the same arguments as in Appendix~\ref{Proof2}, the random variable $\mb{\wt L}_\mb{\mb{\wt B}}$ is semicircular and free from $\mb{\ob L}_\mb{\mb{\wt B}}$ over $\mc{D}_{\wt n}$.

\section{Proof of Proposition~\ref{prop5}} \label{Proof5}
Based on Proposition~\ref{prop4}, the random variable $\mb{\wt L}_{\mb{\wt B}}$ is semicircular and free from $\mb{\ob L}_{\mb{\wt B}}$ over $\mc{D}_{\wt n}$. Thus, according to Appendix~\ref{Proof3}, the $R$-transform of $\mb{\wt L}_{\mb{\wt B}}$ reduces to the covariance function of $\wt{\mb{L}}_{\mb{\wt B}}$ over $\mc{\wt D}$ parameterized by $\mb K_{\wt n}$, i.e.,
\begin{align} \label{R2}
	\mc{R}_{\widetilde{\mb{L}}_{\mb{\wt B}}}^{\mc{\wt D}}(\mb{K}_{\wt n}) &= \mbb{E}_{\mc{D}_{\wt n}}\left[\widetilde{\mb{L}}_{\mb{\wt B}}\mb{K}_{\wt n}\widetilde{\mb{L}}_{\mb{\wt B}}\right]\nonumber\\
	&=
	\begin{bmatrix}
		\sum_{i=1}^{K}\bm{\zeta}_{i}({{\mb{D}}}_{i}) & \mb{0}_{M} & \mb{0}_{M\times N_{\rm{tot}}} &   \\
		\mb{0}_{M} & \mb{0}_{M} & \mb{0}_{M\times N_{\rm{tot}}} &\ \\
		\mb{0}_{N_{\rm{tot}}\times M} & \mb{0}_{N_{\rm{tot}}\times M} & \bm{\wt \zeta}(\mb{\wt D})  \\
	\end{bmatrix},
\end{align}
where $\bm{\wt \zeta}(\mb{\wt D})={\rm{blkdiag}}\left(\wt{\eta}_{1l}(\mb{I}_{M}, \mb{\wt D}, \mb{P}_1),\ldots, \wt{\eta}_{Kl}(\mb{I}_{M},\mb{\wt D},\mb{P}_K)\right)$ and $\bm{\zeta}_{i}({{\mb{D}_{i}}})={\eta}_{il}(\mb{I}_{M},\mb{D}_{i},\mb{P}_i)$.

By the same matrix partitioning as in (\ref{K2}), $\mc{G}_{\mb{L}}^{\mc{D}}(\mb{\Lambda}(z))$ is partitioned into
\begin{align} \label{G2}
	\mc{G}_{\mb{L}}^{\mc{D}}(\mb{\Lambda}(z))=
	\begin{bmatrix}
		\mc{G}_{\mb{\wt D}}(z) & \mb{0}_{M} & \mb{0}_{M\times N_{\rm{tot}}} &\\
		\mb{0}_{M} & \mb{0}_{M} & \mb{0}_{M\times N_{\rm{tot}}} & \\
		\mb{0}_{N_{\rm{tot}}\times M} & \mb{0}_{N_{\rm{tot}}\times M} & \mc{G}_{\mb{D}}(z) & \\
	\end{bmatrix},
\end{align}
where $\mc{G}_{\mb{\wt D}}(z)$ is a $M\times M$ matrix and $\mc{G}_{\mb{D}}(z)$ is a $N_{\rm{tot}}\times N_{\rm{tot}}$ diagonal block matrix with $\{\mc{G}_{\mb{D}_i}(z)\}_{1\leq i \leq K}$ on the diagonal. 

By replacing $\mb{K}$ in (\ref{R2}) with $\mc{G}_{\mb{L}}^{\mc{D}}(\mb{\Lambda}(z))$ in (\ref{G2}) and according to the subordination formula (\ref{sub2}), we have 
\begin{equation} \label{equGL}
	\begin{aligned} {}
		\mc{G}_{\mb{L}}^{\mc{D}}(\mb{\Lambda}(z)) &=\mbb{E}_{\mc{D}}
		\left[\begin{array}  {c:ccc}
			\mb{\Phi}(z) & -(\mb{\wt{C}}_{l}^{\prime})^{\frac{1}{2}}  & -\mb{\ob{H}}_l\mb{P} & \\ \hdashline
			-(\mb{\wt{C}}_{l}^{\prime})^{\frac{1}{2}} & \mb{I}_{M} & \mb{0} & \\
			-\mb{P}^{\dag}\mb{\ob{H}}_l^{\dag} & \mb{0} & \mb{\wt \Phi}(z) &
		\end{array}\right]^{-1},
	\end{aligned}
\end{equation}
where $\mb{\Phi}(z)=z\mb{I}_{M}-\sum_{i=1}^{K}\bm{\zeta}_{i}({{\mc{G}_{\mb{D}_{i}}(z)}})$ and $\mb{\wt \Phi}(z)=\mb{I}_{N_{\rm{tot}}}-\bm{\wt \zeta}(\mb{\mc{G}_{\mb{\wt D}}(z)})$.

By applying the block matrix inversion lemma to (\ref{equGL}), we can obtain the expressions of $\mc{G}_{\mb{\wt D}}(z)$ and $\mc{G}_{\mb{D}}(z)$ as in (\ref{equGD}) and (\ref{equGDT}).

\section{Derivative of Cauchy Transform} \label{Appderivative}
Recall that the Cauchy transform $\mc{G}^\prime_{\mb{B}_l}(z)=\frac{1}{N_t}\mbb{E}\{{\rm{Tr}}(z\mb{I}_{N_t}-\mb{B}_l)^{-2}\}$ in~(\ref{matrixA}) denotes the derivative of $\mc{G}_{\mb{B}_l}(z)$ and according to Proposition~\ref{prop3}, we have
\begin{equation}
	\mc{G}^\prime_{\mb{B}_l}(z) = \frac{1}{N_t}\operatorname{Tr}\left(\mc{G}^\prime_{\mb{D}_1}(z)\right),
\end{equation}
where $\mc{G}^\prime_{\mb{D}_1}(z)$ is the derivative of $\mc{G}_{\mb{D}_1}(z)$, which satisfies the following matrix-valued fix-point equations
\begin{align}
	\mc{G}^\prime_{\mb{D}_1}(z) &= \mc{G}_{\mb{D}_1}(z)\left[\mc{G}^{-1}_{\mb{D}_1}(z)\right]^\prime\mc{G}_{\mb{D}_1}(z), \\
	\mc{G}^\prime_{\mb{D}_2}(z) &= \mc{G}_{\mb{D}_2}(z)\left[\mc{G}^{-1}_{\mb{D}_2}(z)\right]^\prime\mc{G}_{\mb{D}_2}(z),
\end{align}
where $\left[\mc{G}^{-1}_{\mb{D}_1}(z)\right]^\prime$ and $\left[\mc{G}^{-1}_{\mb{D}_2}(z)\right]^\prime$ are respectively the derivatives of $\mc{G}^{-1}_{\mb{D}_1}(z)$ and $\mc{G}^{-1}_{\mb{D}_2}(z)$. And some additional parameters are given as follows
\begin{align}
	\left[\mc{G}^{-1}_{\mb{D}_1}(z)\right]^\prime &= \mb{\Psi}^\prime(z)-\mb{\ob{G}}_l^{\dag}\left[\mb{\wt \Psi}^{-1}(z)\right]^\prime\mb{\ob{G}}_l, \\
	\left[\mc{G}^{-1}_{\mb{D}_2}(z)\right]^\prime &= \mb{\wt \Psi}^\prime(z)-\mb{\ob G}_l\left[{\mb{\Psi}}^{-1}(z)\right]^\prime\mb{\ob G}_l^{\dag}, \\
	\left[{\mb{\Psi}}^{-1}(z)\right]^\prime &= -{\mb{\Psi}}^{-1}(z)\left(\mb{I}_{N_t} - \wt{\delta}^\prime(\mc{G}_{\mb{D}_2}(z))\right){\mb{\Psi}}^{-1}(z),\\
	\left[{\mb{\wt \Psi}}^{-1}(z)\right]^\prime &= {\mb{\wt \Psi}}^{-1}(z)\left(\sum_{i=1}^{K}\delta^\prime_i(\mc{G}_{\mb{D}_{1i}}(z))\right){\mb{\wt \Psi}}^{-1}(z),\\
	\wt{\delta}^\prime(\mc{G}_{\mb{D}_2}(z)) &= \wt{\delta}(\mc{G}^\prime_{\mb{D}_2}(z)), \quad
	\delta^\prime_i(\mc{G}_{\mb{D}_{1i}}(z)) = 
	\delta_i(\mc{G}^\prime_{\mb{D}_{1i}}(z)).
\end{align}

As for the Cauchy transform $\mc{G}^\prime_{\mb{\wt{B}}_l, \mb{\Xi}}(z)=\frac{1}{M_l}\mbb{E}\{{\rm{Tr}}((z\mb{I}_{M_l}-\mb{\wt B}_l)^{-2}\mb{\Xi})\}$ in~(\ref{matrixY}), we have 
\begin{equation}
	\mc{G}^\prime_{\mb{\wt{B}}_l, \mb{\Xi}}(z) = \frac{1}{M_l}\operatorname{Tr}(\mc{G}^\prime_{\mb{\wt D}}(z)\mb{\Xi}),
\end{equation}
where $\mc{G}^\prime_{\mb{\wt D}}(z)$ is the derivative of $\mc{G}_{\mb{\wt D}}(z)$, which satisfies the following matrix-valued fix-point equations
\begin{align}
	\mc{G}^\prime_{\mb{\wt D}}(z) &= \mc{G}_{\mb{\wt D}}(z)\left[\mc{G}^{-1}_{\mb{\wt D}}(z)\right]^\prime\mc{G}_{\mb{\wt D}}(z),\\
	\mc{G}^\prime_{\mb{D}}(z) &= \mc{G}_{\mb{D}}(z) \left[\mc{G}^{-1}_{\mb{D}}(z)\right]^\prime \mc{G}_{\mb{D}}(z), 
\end{align}
where $\left[\mc{G}^{-1}_{\mb{\wt D}}(z)\right]^\prime$ and $\left[\mc{G}^{-1}_{\mb{D}}(z)\right]^\prime$ are respectively the derivatives of $\mc{G}^{-1}_{\mb{\wt D}}(z)$ and $\mc{G}^{-1}_{\mb{D}}(z)$. And some additional parameters are given as follows
\begin{align}
	\left[\mc{G}^{-1}_{\mb{\wt D}}(z)\right]^\prime &= \mb{\Phi}^\prime(z) -  \mb{\ob{H}}_l\mb{P}\left[\mb{\wt \Phi}^{-1}(z)\right]^\prime\mb{P}^{\dag}\mb{\ob{H}}_l^{\dag},\\
	\left[\mc{G}^{-1}_{\mb{D}}(z)\right]^\prime &= \mb{\wt \Phi}^\prime(z)-\mb{P}^{\dag}\mb{\ob{H}}_l^{\dag}\left[\mb{\Phi}^{-1}(z)\right]^\prime\mb{\ob{H}}_l\mb{P}-\mb{\Upsilon}^\prime(z), \\
	\left[\mb{\Phi}^{-1}(z)\right]^\prime &= -\mb{\Phi}^{-1}(z)\left(\mb{I}_{M_l}-\sum_{i=1}^{K}\bm{\zeta}^\prime_{i}({{\mc{G}_{\mb{D}_{i}}(z)}})\right)\mb{\Phi}^{-1}(z), \\
	\left[\mb{\wt \Phi}^{-1}(z)\right]^\prime &= \mb{\wt \Phi}^{-1}(z)\left(\bm{\wt \zeta}^\prime(\mb{\mc{G}_{\mb{\wt D}}(z)})\right)\mb{\wt \Phi}^{-1}(z),\\
	\bm{\zeta}^\prime_{i}({{\mc{G}_{\mb{D}_{i}}(z)}}) &= \bm{\zeta}_{i}({{\mc{G}^\prime_{\mb{D}_{i}}(z)}}),\quad
	\bm{\wt \zeta}^\prime(\mb{\mc{G}_{\mb{\wt D}}(z)}) = \bm{\wt \zeta}(\mb{\mc{G}^\prime_{\mb{\wt D}}(z)}).
\end{align}
In addition, $\mb{\Upsilon}^\prime(z)$ is the derivative of $\mb{\Upsilon}(z)$, which can be easily obtained by the matrix calculus.


\end{appendices}

\bibliographystyle{IEEEtran}

\end{document}